\documentclass{article}
\usepackage{setspace}
\usepackage[hidelinks]{hyperref}
\usepackage{amsmath}
\usepackage{parskip} 
\usepackage{natbib}
\usepackage{mathrsfs}
\usepackage{amsthm}

\newtheorem{lemma}{Lemma}[section]

\usepackage{bm}
\usepackage{booktabs}
\usepackage{graphics}
\usepackage{multicol}
\usepackage[bf]{caption}
\usepackage{amssymb}
\usepackage{multirow}
\usepackage{booktabs}
\usepackage{subcaption}

\usepackage{comment}
\usepackage[margin=1in]{geometry}
\usepackage{floatrow}
\newfloatcommand{capbtabbox}{table}[][\FBwidth]
\usepackage[titletoc]{appendix}
\usepackage[graphicx]{realboxes}
\usepackage{gensymb}
\usepackage{color}

\usepackage{mathtools}
\usepackage{siunitx}
\usepackage{adjustbox}
\usepackage{multirow}
\usepackage{floatrow}
\newcommand{\MATLAB}{\textsc{Matlab}\xspace}
\usepackage{xspace}
\usepackage{listings}
\usepackage{color} 
\definecolor{mygreen}{RGB}{28,172,0} 
\definecolor{mylilas}{RGB}{170,55,241}

\title{PP-wave reflection coefficient for vertically cracked media: Single set of aligned cracks}
\author{Filip P. Adamus\footnote{
Department of Earth Sciences, Memorial University of Newfoundland, Canada, {\tt adamusfp@gmail.com}}\,\,\,}

\date{}

\begin{document}
\maketitle
\begin{abstract}
The main goal of this paper is to analyse the influence of cracks on the azimuthal variations of amplitude.
We restrict our investigation to a single set of vertical, circular, and flat cavities aligned along a horizontal axis.
Such cracks are embedded in either isotropic surroundings or transversely isotropic background with a vertical symmetry axis.
We employ the effective medium theory to obtain either transversely-isotropic material with a horizontal symmetry axis or an orthotropic medium, respectively.
To consider the amplitudes, we focus on a Vavry\v{c}uk-P\v{s}en\v{c}\'{i}k  approximation of the PP-wave reflection coefficient. 
We assume that cracks are situated  in one of the halfspaces being in welded contact.
Azimuthal variations depend on the background stiffnesses, incidence angle, and crack density parameter.
Upon analytical analysis, we indicate which factors (such as background's saturation) cause the reflection coefficient to have maximum absolute value in the direction parallel or perpendicular to cracks. 
We discuss the irregular cases, where such extreme values appear in the other than the aforementioned directions.
Due to the support of numerical simulations, we propose graphic patterns of two-dimensional amplitude variations with azimuth.
The patterns consist of a series of shapes that change with the increasing value of the crack density parameter.
Schemes appear to differ depending on the incidence angle and the saturation.
Finally, we extract these shapes that are characteristic of gas-bearing rocks. 
They may be treated as gas indicators.
We support the findings and verify our patterns using real values of stiffnesses extracted from the sedimentary rocks' samples.    
\end{abstract}
\section{Introduction}
Recently, the amplitude variations with azimuth (AVA) became a topic of interest for many seismologists. Such physical phenomena occur due to either intrinsic or induced azimuthal anisotropy of the medium. 
The\newpage latter kind is caused by thin layering or cracks embedded in the azimuthally-independent background, under the condition that at least one inhomogeneity is not parallel to the reference plane.
The properties of cracks that induce AVA are essential from the exploration point of view since cracks may cause fluids, such as gas, to flow.
Therefore, in this paper, we investigate the above-mentioned induced variations.
We refer to them, in a singular form, as CAVA to emphasise that AVA is caused by cracks---not the layering or intrinsic anisotropy.
To analyse such variations, we utilise effective medium theory and reflection coefficients.

The idea of homogenisation of elastic media containing inclusions dates back to~\citet{Bruggeman} and~\citet{Eshelby}. A good summary of the micromechanics researchers' efforts can be found in~\citet{Kachanov}.  
From a strictly geophysical perspective, the elastics of cracked media has been studied and popularised by~\citet{SchDouma},~\citet{SchSayers}, and~\citet{SchHelbig}.
They treat cracks as infinitely weak and thin planes embedded in the background medium.

For over a century, researchers have been interested in finding the formulation of reflection and transmission coefficients at an interface between two elastic halfspaces.
The exact solution to the plane wave reflection and transmission problem for isotropic media was shown by~\citet{Zoeppritz}.
An elegant extension of the explicit solution to monoclinic halfspaces was presented by~\citet{SchProtazio}.
Due to the complexity of analytical formulations for both isotropic and anisotropic cases, researchers focused on various reflection and transmission coefficients' approximations.
For the azimuthally-independent approximations, readers may refer to~\citet{AVO}.
If reflection and transmission coefficients do change with azimuth, the approximations become more complicated.
\citet{Ruger} proposed formulations for two transversely-isotropic halfspaces with a horizontal axis of symmetry.
However, his derivations appear inaccurate if lower symmetries are considered. 
Halfspaces with an arbitrary symmetry were discussed by~\citet{UrsinHaugen},~\citet{ZillmerEtAl}, and~\citet{VavrycukPsencik}.
To obtain the approximations,~\citet{UrsinHaugen} assumed weak elastic contrast at the interface only. 
On the other hand,~\citet{ZillmerEtAl} allowed the contrast to be strong but assumed weak anisotropy.
Both approximations are very complicated and lengthy.
More user-friendly derivations were shown by~\citet{VavrycukPsencik}, who assumed both weak contrast interface and weak anisotropy.
Upon introducing further linearisation, their approximations can be reduced to an elegant formulation shown by~\citet{PsencikMartins}, and then, it can be simplified to the aforementioned, popular approximation of~\citet{Ruger}.
Due to the convenience of analytical analysis and the accuracy of the approximation for low symmetry classes, we focus on the reflection coefficient estimated by~\citet{VavrycukPsencik}.
Further, we restrict our analysis to PP-plane waves.

Numerous authors have employed the effective medium theory in the context of azimuthal variations of the reflection coefficient.
Specifically, they often have used PP-amplitudes to predict the background and fracture parameters.
Thus, they have focused on the solution of the inverse problem.
However, due to a large number of unknowns, the azimuthal inversion becomes a difficult task.
Some authors have considered isotropic background and the aforementioned R{\"u}ger's equations~\citep[e.g.,][]{RugerGray}.
The others considered more sophisticated approximations, employing novel techniques to reduce the number of parameters, but still assuming an isotropic background~(\citet{RppInv}~and~\citet{derivative}).
In contrast to the authors mentioned above,~\citet{RppInvORT} have allowed the background to present lower than isotropic symmetry.
However, they have utilised an additional linearisation of the reflection coefficient.  
In this paper, we do not focus on the explicit inversion of background and fracture parameters.
This way, we can consider an anisotropic background and Vavry\v{c}uk-P\v{s}en\v{c}\'{i}k  approximation, with no additional linearisations or assumptions.
Using analytical and numerical methods, we try to better understand the nature of azimuthal variations of amplitude for cracked media.
Specifically, we analyse the shapes of these variations.
Thorough investigation allows us to notice not only the ellipse or peanut shapes, as assumed or indicated by the other authors~\citep[e.g.,][]{derivative}.
Further, some shapes occur to be more probable for specific saturations, incidences, and crack concentrations.
Therefore, we also touch on the inverse problem; however, in an indirect, not explicit way.
In other words, instead of focusing on the popular Bayesian framework, we investigate the patterns and characteristic attributes of CAVA.
Principally, we are interested in the shapes that are typical for gas-bearing rocks.
\section{Theory}
\subsection{Elasticity tensor and stability conditions}
Consider a three-dimensional Cartesian coordinate system with $x_i$ axes, where $x_3$ denotes the vertical axis.
An elasticity tensor is a forth-rank Cartesian tensor that relates stress and strain second-rank tensors. 
A material whose elastic properties are rotationally invariant about one symmetry axes is called to be transversely isotropic (TI).
This paper focuses on a TI medium with a rotation symmetry axis that coincides with the $x_3$-axis; we refer to such a medium as a VTI material.
In a Voigt notation, an elasticity tensor of a VTI material is represented by 
\begin{equation}\label{crack:end}
\bm{C}=
\begin{bmatrix}
C_{11}& C_{12} & C_{13} & 0 & 0 & 0\\
 C_{12}& C_{11}& C_{13}& 0 & 0 & 0\\
C_{13}& C_{13}& C_{33} & 0 & 0 & 0 \\
0 & 0 & 0 & C_{44} & 0 & 0\\
0 & 0 & 0 & 0 & C_{44} & 0\\
0 & 0 & 0 & 0 &0 &C_{66}\\
\end{bmatrix}\,,
\end{equation}  
where $C_{12}=C_{11}-2C_{66}$\,.
If additionally $C_{11}=C_{33}$\,, $C_{12}=C_{13}\,$, and $C_{44}=C_{66}$\,, medium becomes isotropic.
An elasticity tensor is physically feasible if it obeys the stability conditions.
These conditions \citep[e.g.,][Section 4.3]{SlawinskiRedNEW} originate from the necessity of expending energy to deform a material. 
This necessity is mathematically expressed by the positive definiteness of the elasticity tensor. 
A tensor is positive definite if and only if all eigenvalues of its $6\times6$ matrix representation are positive. 
For a VTI elasticity tensor, this entails 
\begin{equation}
C_{11}-|C_{12}|>0\,,\quad C_{33}(C_{11}+C_{12})> 2C_{13}^2\,,\quad {\rm and}\quad C_{44}>0\,.
\end{equation}
For isotropic tensor, only one condition is required, namely,
\begin{equation}
\label{ineq}
    C_{11}>\tfrac{4}{3}\,C_{44}>0\,.
\end{equation}
\subsection{Elastic medium with a single ellipsoidal inclusion}
Consider a linearly elastic material of volume $V$ that contains a single region (inclusion) with different elastic properties occupying volume $V_1$\,.
We denote the inhomogeneity surroundings with a subscript $0$\,, whereas $1$ indicates the embedded region. 
The strain of an effective (homogenised) material, averaged over volume $V$\,, can be expressed in terms of the strain average over the background surroundings $\overline{\bm{\varepsilon}}_0$\,, and strain average over the inhomogeneity $\overline{\bm{\varepsilon}}_1$\,, as follows.\begin{equation}\label{one}
\overline{\bm{\varepsilon}}=\frac{V-V_1}{V}\,\overline{\bm{\varepsilon}}_0+\frac{V_1}{V}\,\overline{\bm{\varepsilon}}_1=\phi_0\bm{S}_0\bm{:}\overline{\bm{\sigma}}_0+\phi_1\bm{S}_1\bm{:}\overline{\bm{\sigma}}_1\,,
\end{equation}
where $\bm{\varepsilon}$ stands for the second-rank strain tensor, $\bm{\sigma}$ denotes the second-rank stress tensor, $\bm{S}$ is the fourth-rank compliance tensor (the inverse of an elasticity tensor), and operator $\bm{:}$ is the double-dot product.
Volume fractions are denoted by $\phi$\,.
We assume an external stress at infinity that is uniform, namely,
\begin{equation}
\bm{\sigma}_\infty=\phi_0\overline{\bm{\sigma}}_0+\phi_1\overline{\bm{\sigma}}_1\,.
\end{equation}
The external stress can be related to the inhomogeneity stress solely, by using a fourth-rank stress concentration tensor $\bm{B}$\,, that is,
\begin{equation}\label{three}
\overline{\bm{\sigma}}_1=\bm{B}\bm{:}\bm{\sigma}_\infty\,.
\end{equation}
Combining equations~(\ref{one})--(\ref{three}), we get
\begin{equation}\label{extra}
\overline{\bm{\varepsilon}}=\bm{S}_0\bm{:}\bm{\sigma}_\infty+\phi_1(\bm{S}_1-\bm{S}_0)\bm{:}\bm{B}\bm{:}\bm{\sigma}_\infty=\bm{S}_0\bm{:}\bm{\sigma}_\infty+\Delta\bm{\varepsilon}\,,
\end{equation}
where $\Delta\bm{\varepsilon}$ is the extra strain of the material due to inclusion.
Now, we can define fourth-rank compliance contribution tensor
\begin{equation}\label{h}
\bm{H}:=(\bm{S}_1-\bm{S}_0)\bm{:}\bm{B}
\end{equation}
that is a key tensor since---upon inserting expression~(\ref{h}) into~(\ref{extra})---the contribution of an inhomogeneity can be expressed by $\bm{H}$ and volume fraction only. 
$\bm{H}$ depends on elastic properties of an inclusion and on $\bm{B}$ that, in turn, depends on inclusion's shape.
If we assume the ellipsoidal shape of the inhomogeneity, we can express $\bm{B}$ in terms of fourth-rank Eshelby tensor, $\bm{s}$\,, which leads to a significant simplification of the so-called Eshelby problem~\citep{Eshelby}. 
If different shapes are considered, the stresses and strain inside the inclusion are not constant. 
In turn, the contribution of such inclusions cannot be expressed in terms of Eshelby tensor, and more complicated techniques must be involved~\citep{KS}.
Herein, we assume ellipsoidal shape and get
\begin{equation}\label{b}
\bm{B}=\left[\bm{J}+\bm{Q}\bm{:}(\bm{S}_1-\bm{S}_0)\right]^{-1}\,,
\end{equation}
where $\bm{J}$ is the fourth-rank unit tensor---$\left(\delta_{ik}\delta_{j\ell}+\delta_{i\ell}\delta_{jk}\right)/2$\,, where $\delta$ is the Kronecker delta---and $\bm{Q}$ is the fourth-rank tensor related to Eshelby tensor, $\bm{s}$\,, namely,
\begin{equation}\label{q}
\bm{Q}=\bm{C}_0\bm{:}(\bm{J}-\bm{s})\,.
\end{equation}
$\bm{C}_0$ is the fourth-rank elasticity tensor of the background material.
Upon inserting expression~(\ref{b}) into~(\ref{h}), we obtain
\begin{equation}
\bm{H}=\left[\left(\bm{S}_1-\bm{S}_0\right)^{-1}+\bm{Q}\right]^{-1}\,.
\end{equation}
In the next section, we refer to the above derivations to discuss the contribution of multiple, circular, and flat cavities (dry cracks) embedded in a VTI background. 
\subsection{Elastic VTI medium with circular cracks}
If an embedded, single region is a flat (planar) crack, then the extra strain from equation~(\ref{extra}) can be written as
\begin{equation}\label{cracks}
\Delta\bm{\varepsilon}=\frac{S}{2V}\left(\bm{bn}+\bm{nb}\right)\,,
\end{equation} 
where $\bm{b}$ is the average---over crack surface $S$---displacement discontinuity vector, and $\bm{n}$ is the crack normal vector. In our notation, $\bm{bn}$ and $\bm{nb}$ are the outer products.
Assuming linear displacements, we can introduce a second-rank crack compliance tensor $\bm{Z}$\,, namely,
\begin{equation}\label{Z}
\frac{S}{V}\,\bm{b}=\bm{n}\cdot\bm{\sigma}_\infty\cdot\bm{Z}\,.
\end{equation}
Upon inserting expression~(\ref{Z}) into~(\ref{cracks}) and comparing the result to the extra strain from equation~(\ref{extra}), we get an equality
\begin{equation}
\bm{nZn}=\phi_1\bm{H}\,.
\end{equation}
We see that tensor $\bm{Z}$ is akin to ''fracture system compliance tensor" shown in celebrated papers of~\citet{SchSayers} and~\citet{SchHelbig}.
Since $\bm{Z}$ is symmetric, three principal directions of the crack compliance must exist.
Therefore, we can write
\begin{equation}
\bm{Z}=Z_N\bm{nn}+Z_{tt}\bm{tt}+Z_{ss}\bm{ss}\,,
\end{equation}
where $\bm{n}$\,, $\bm{t}$\,, and $\bm{s}$ are mutually orthogonal vectors.
If a crack is circular, it becomes rotationally invariant, so that $Z_{tt}=Z_{ss}=:Z_T$\,.
Since $\bm{nn}+\bm{tt}+\bm{ss}=\bm{I}$\,, we get
\begin{equation}\label{znzt}
\bm{Z}=Z_N\bm{n\,n}+Z_T(\bm{I}-\bm{n\,n})\,,
\end{equation}
where $\bm{n}$ is the crack normal.
Given the assumptions above, $Z_N$ and $Z_T$ completely determine the elastic contribution of a circular crack.
To obtain them, we need to compute $\bm{H}$ that depends on crack stiffnesses and Eshelby tensor.
In turn, Eshelby tensor depends on background stiffnesses and crack shape. 
Complicated computation of $\bm{s}$ for different symmetry classes of the background medium and shapes of inclusion is well-explained in \citet{Sev2005} and \citet{KS}. 
Let us derive $Z_N$ and $Z_T$ for \linebreak a VTI background.
If we assume that a circular crack is dry, meaning that its elasticity parameters are zero, we get
\begin{equation}\label{Zn}
Z_N=\frac{8c_3e}{3c_{1}\left(1-\dfrac{c_{13}^2}{c_1^2}\right)}\geq0\,,
\end{equation}
\begin{equation}\label{Zt}
Z_T=\frac{16e}{3c_{44}\left(c_2+c_3-c_4\right)}\geq0\,,
\end{equation}
where
\begin{equation}
\begin{gathered}
c_1:=\sqrt{c_{11}c_{33}}\,,\qquad c_2:=\sqrt{\frac{c_{66}}{c_{44}}}\,,\\
c_3:=\sqrt{\frac{(c_1-c_{13})(c_1+c_{13}+2c_{44})}{c_{33}c_{44}}}\,,\qquad
c_4:=\frac{2c_{44}c_{3}}{c_1+c_{13}+2c_{44}}\,.
\end{gathered}
\end{equation}
Throughout the paper, $c_{ij}$ denote the stiffnesses of a background medium, expressed in a Voigt notation.
Parameter
\begin{equation}\label{e}
e=\frac{ma^3}{V}\,,
\end{equation}
is the crack density with crack ratio $a$ and number of cracks $m$ (in this, single crack case, $m=1$).
The above results are identical to the ones of~\citet{Guo}.
If the background is isotropic, then expressions~(\ref{Zn}) and~(\ref{Zt}) reduce to
\begin{equation}\label{Zniso}
Z_N=\frac{4c_{11}e}{3c_{44}(c_{11}-c_{44})}\geq0\,,
\end{equation}
\begin{equation}\label{Ztiso}
Z_T=\frac{16c_{11}e}{3c_{44}(3c_{11}-2c_{44})}\geq0\,.
\end{equation}
Expressions~(\ref{Zniso}) and~(\ref{Ztiso}) were derived previously by numerous authors~\citep[e.g.,][]{Hudson80}.
To satisfy stability conditions---for dry circular cracks---$Z_N$ and $Z_T$ must be nonnegative, no matter if the background is VTI or isotropic.
 
So far, we have discussed a case of a single inhomogeneity surrounded by the background medium.
If there are multiple flat cracks, the strain equation~(\ref{extra}) can be generalised to
\begin{equation}\label{extra2}
\overline{\bm{\varepsilon}}=(\bm{S}_0+\Delta\bm{S})\bm{:}\bm{\sigma}_\infty\,,
\end{equation}
where
\begin{equation}\label{deltaS}
\Delta\bm{S}=\sum_{k=1}^m\phi_{k}\bm{H}_{(k)}=\sum_{k=1}^m\bm{n}_{(k)}\bm{Z}_{(k)}\bm{n}_{(k)}
\end{equation}
with $m$ being the total number of cracks.
Tensors $\bm{Z}_{(k)}$ apart from depending on the shape and properties of cracks, they also may depend on the interactions between the inhomogeneities. 
In this paper, however, we assume the non-interaction approximation (NIA), so that cracks are treated as they were isolated, and $\bm{Z}_{(k)}$ can be obtained in a manner discussed above.
If cracks have the same shape and orientation, $Z_N$ and $Z_T$ for each inhomogeneity do sum up, and the concentration of cracks is reflected in the parameter $e$ with $m>1$ (see expression~(\ref{e})).
The NIA is particularly useful for strongly oblate or planar cracks due to its good accuracy even for higher values of density parameter~\citep{GrechkaKachanov}. 
For flat cracks, the shape factors such as roughness can be ignored~\citep{KS}.
In the next section, we discuss a particular case of expression~(\ref{deltaS}).
%
\subsection{Elastic VTI medium with a single set of aligned vertical cracks}
Consider one set of flat cracks having identical circular shapes that are embedded in a VTI background medium. 
Assume that all cracks are dry and have the same orientation.
Upon combining expressions~(\ref{znzt})--(\ref{Zt}), we can rewrite expression~(\ref{deltaS}) as
\begin{equation}\label{crackss}
\Delta\bm{S}=\bm{n}\left[\frac{8c_3e}{3c_{1}\left(1-\dfrac{c_{13}^2}{c_1^2}\right)}\bm{nn}+\frac{16e}{3c_{44}\left(c_2+c_3-c_4\right)}\left(\bm{I}-\bm{nn}\right)\right]\bm{n}
\end{equation}
where $\bm{n}$ is a normal to the set of cracks and $e$ describes the crack concentration with $m>1$\,.
Having the above expression, we can obtain the effective elasticity tensor of a homogenised medium. 
Since, elasticity is the inverse of the compliance tensor, we need to examine
\begin{equation}\label{crack:end}
\bm{C}^{\rm{eff}}=\left(\bm{S}_0+\Delta\bm{S}\right)^{-1}\,.
\end{equation}
If a set of cracks is vertical, then effective elasticity tensor has monoclinic or higher symmetry~\citep{SchEtAl}.
We propose to consider a simpler case, where cracks have a normal parallel to the $x_1$-axis.
Then, the tensor exhibits at least orthotropic symmetry.
In a Voigt notation, the effective elasticity tensor has the following form.
\begin{gather}\label{crack:end}
\setlength{\arraycolsep}{2pt}
\renewcommand*{\arraystretch}{1.1}
\bm{C}^{\rm{eff}}=
\begin{bmatrix}
c_{11}(1-\delta_N) & c_{12}(1-\delta_N) & c_{13}(1-\delta_N) & 0 & 0 & 0\\
 c_{12}(1-\delta_N) & c_{11}\left(1-\delta_N\dfrac{c_{12}^2}{c_{11}^2}\right) & c_{13}\left(1-\delta_N\dfrac{c_{12}}{c_{11}}\right) & 0 & 0 & 0\\
c_{13}(1-\delta_N) & c_{13}\left(1-\delta_N\dfrac{c_{12}}{c_{11}}\right) & c_{33}\left(1-\delta_N\dfrac{c_{13}}{c_{11}c_{33}}\right) & 0 & 0 & 0 \\
0 & 0 & 0 & c_{44} & 0 & 0\\
0 & 0 & 0 & 0 & c_{44}\left(1-\delta_{T_1}\right) & 0\\
0 & 0 & 0 & 0 &0 &c_{66}\left(1-\delta_{T_2}\right)
\end{bmatrix}\,,
\raisetag{-0.03\baselineskip}
\end{gather}  
\normalsize
where $c_{12}=c_{11}-2c_{66}$\,.
Deltas are the combinations of crack compliances and background stiffnesses, namely,
\begin{equation}
\delta_N:=\frac{Z_Nc_{11}}{1+Z_Nc_{11}}
=
\frac{8c_{11}c_3e}{8c_{11}c_3e+3c_1\left(1-\dfrac{c_{13}^2}{c_1^2}\right)}
\,,
\end{equation}
\begin{equation}
\delta_{T_1}:=\frac{Z_Tc_{44}}{1+Z_Tc_{44}}
=
\frac{16e}{16e+3(c_2+c_3-c_4)}
\,,
\end{equation}
\begin{equation}
\delta_{T_2}:=\frac{Z_Tc_{66}}{1+Z_Tc_{66}}
=
\frac{16c_{66}e}{16c_{66}e+3c_{44}(c_2+c_3-c_4)}
\,.
\end{equation}
Note that if there are no cracks, the effective tensor reduces to the background; thus, it has five independent stiffnesses, instead of nine.
Dry cracks embedded in the background always increase certain compliances of the effective medium ($Z_N$ and $Z_T$ must be positive).
Therefore, the effective elastic properties are weaker than the properties of the background.
Analogously, effective stiffnesses can be derived for the isotropic background.

So far, in this section, we have discussed the limiting case of dry, circular, flat cracks.
In other words, we have considered the ellipsoids with one pair of the semi-axes of equal length and the third tending to zero $(a_1=a_2=a, \,a_3\rightarrow0)$\,.
The results of this section also can be a good approximation for strongly oblate spheroids (penny-shaped cavities).
In the case of such shape, one semi-axis is much smaller than the other two ($a_3\ll a$)\,.
As shown by~\citet{Sev2005}, the values of Eshelby tensor---compared to the case of $a_3\rightarrow0$\,---do not change significantly, which means that expression~(\ref{crackss}) is accurate enough.
Note that the volume fraction $\phi_i$ or aspect ratio ($a_3/a$) of penny-shaped cavities are very small so that they are irrelevant for their characterisation~\citep[e.g.,][]{Kachanov1994}.
The only useful concentration parameter is the crack density, whose values are limited by the accuracy of the NIA only.  
\subsection{Vavry\v{c}uk-P\v{s}en\v{c}\'{i}k approximation}
Consider a spherical coordinate system.
Vavry\v{c}uk-P\v{s}en\v{c}\'{i}k approximation for the PP-wave reflection coefficient---valid for at least monoclinic halfspaces being in welded contact---is
\small
\begin{gather}\label{vav}
\begin{aligned}
&R_{\rm{pp}}(\theta,\psi)= R_{\rm{ipp}}(\theta)+\frac{1}{2}\sin^2\theta\bigg\{
\left[\Delta\left(\frac{C_{23}+2C_{44}-C_{33}}{C_{33}}\right)-8\Delta\left(\frac{C_{44}-C_{55}}{2C_{33}}\right)\right]\sin^2\psi
\\
&+\Delta\left(\frac{C_{13}+2C_{55}-C_{33}}{C_{33}}\right)\cos^2\psi
+2\left[\Delta\left(\frac{C_{36}+2C_{45}}{C_{33}}\right)-4\Delta\left(\frac{C_{45}}{C_{33}}\right)\right]\sin\psi\cos\psi
\bigg\}
\\
&+\frac{1}{2}\sin^2\theta\tan^2\theta\bigg\{\Delta\left(\frac{C_{22}-C_{33}}{2C_{33}}\right)\sin^4\psi+\Delta\left(\frac{C_{11}-C_{33}}{2C_{33}}\right)\cos^4\psi
\\
&+\Delta\left(\frac{C_{12}+2C_{66}-C_{33}}{C_{33}}\right)\sin^2\psi\cos^2\psi+\Delta\left(\frac{C_{26}}{C_{33}}\right)\sin^3\psi\cos\psi+\Delta\left(\frac{C_{16}}{C_{33}}\right)\cos^3\psi\sin\psi
\bigg\}\,,\\ 
\end{aligned}
\raisetag{-0.03\baselineskip}
\end{gather}
\normalsize
where $\theta$ is the incidence angle measured from the $x_3$-axis and $\psi$ is the azimuthal angle measured from the $x_1$-axis towards the $x_2$-axis. $\Delta$ stands for the difference between elastic effective parameters of lower and upper halfspaces, respectively ($\Delta=w^\ell-w^u$\,, where $w$ is some parameter).
First term of the above approximation denotes the PP reflection coefficient between two slightly different isotropic media, proposed by~\citet{AkiRichards}. The rest of the terms in expression~(\ref{vav}) are the correction terms due to the anisotropy of the halfspaces.
The reflection coefficient of the isotropic part is
\begin{equation}
R_{\rm{ipp}}(\theta)=\frac{1}{2}\frac{\Delta \left(\rho\alpha\right)}{\overline{\rho\alpha}}+\frac{1}{2}\frac{\Delta\alpha}{\overline{\alpha}}\tan^2\theta-2\,\frac{\Delta\left(\rho\beta^2\right)}{\overline{\rho\alpha}^2}\sin^2\theta\,,
\end{equation}
where $\rho$ is the mass density, whereas $\alpha$ and $\beta$ are P and S wave velocities of the isotropic medium that can be chosen arbitrarily. 
We follow~\citet{VavrycukPsencik} who have defined $\alpha=\sqrt{C_{33}/\rho}$\,\,and\,$\beta=\sqrt{C_{55}/\rho}$\,.
The bar stands for the average properties between two halfspaces, for instance, $\overline{\alpha}=(\alpha^\ell+\alpha^u)/2$\,.
%
\section{CAVA conjecture}
In this section, we analyse the effect of a single set of dry, circular cracks on azimuthal variations of amplitude. 
We assume that cracks are vertical with a normal parallel to the $x_1$-axis, the background is VTI (in some cases isotropic), and we use Vavry\v{c}uk-P\v{s}en\v{c}\'{i}k  approximation.
Hence, to obtain the reflection coefficient, we insert effective elasticity parameters from matrix~(\ref{crack:end}) into expression~(\ref{vav}).
CAVA depends on the crack density parameter $e$\,, incidence angle, and  background stiffnesses.
To better understand and separate the effect of $e$ on azimuthal variations---at the end of the section---we fix $\theta$ and the elasticity parameters.
By doing so, we obtain $R_{pp}(\psi,e)$ that corresponds to various seismological situations.
We present $R_{pp}(\psi,e)$ as a series of two-dimensional polar graphs that change with increasing concentration of cracks. 
We conjecture what the most probable patterns of such series are.

It is useful to prove that $R_{pp}(\psi)$ expressed in polar coordinates has at least two-fold rotational symmetry.
In this way, one quadrant instead of entire graph may be examined, which facilitates the analysis.
Throughout the paper, we focus on the quadrant, where $\psi\in[0^\circ,90^\circ]$\,.
Consider following Lemma.
\begin{lemma}
The graph of PP reflection coefficient, $R_{pp}(\psi)$\,, computed for orthotropic medium using Vavry\v{c}uk-P\v{s}en\v{c}\'{i}k  approximation has at least two-fold symmetry, where $R_{pp}(\psi)$ is expressed in polar coordinates.
\end{lemma}
\begin{proof}
The equality,
\begin{equation*}
R_{pp}(-\psi)=R_{ipp}+a_1\sin^2\psi+a_2\cos^2\psi+a_3\sin^4\psi+a_4\cos^4\psi+a_5\sin^2\psi\cos^2\psi=R_{pp}(\psi)\,,
\end{equation*}
where $a_i$ are constants, means that the graph is symmetric about polar $x_1$-axis that coincides with $\psi=0^\circ$\,.
Also,
\begin{equation*}
R_{pp}(\psi+180^\circ)=R_{pp}(\psi)\,;
\end{equation*}
the graph is symmetric with respect to the origin.
The above symmetries imply the symmetry about the $x_2$-axis that coincides with $\psi=90^\circ\,$.
Hence, the graph has at least two-fold symmetry and is represented by four identical quadrants.
\end{proof}

There are numerous possible shapes of azimuthal variations.
In general, we can distinguish two main kinds of CAVA graphs---regular or irregular.
Former one, assures that the pair of minimal and maximal value of $R_{pp}(\psi)$---where $\forall\psi\in[0^\circ,90^\circ]$---is obtained for the pair of azimuths $\psi=0^\circ$ and $\psi=90^\circ$\,.
The irregularity occurs if $\min R_{pp}(\psi)$ or $\max R_{pp}(\psi)$ is given for other azimuths.
Examples of both kinds of shapes---that we discuss in the next sections---are shown in Figure~\ref{fig:shapes}.
\begin{figure}[ht!]
\includegraphics[width=0.8\textwidth]{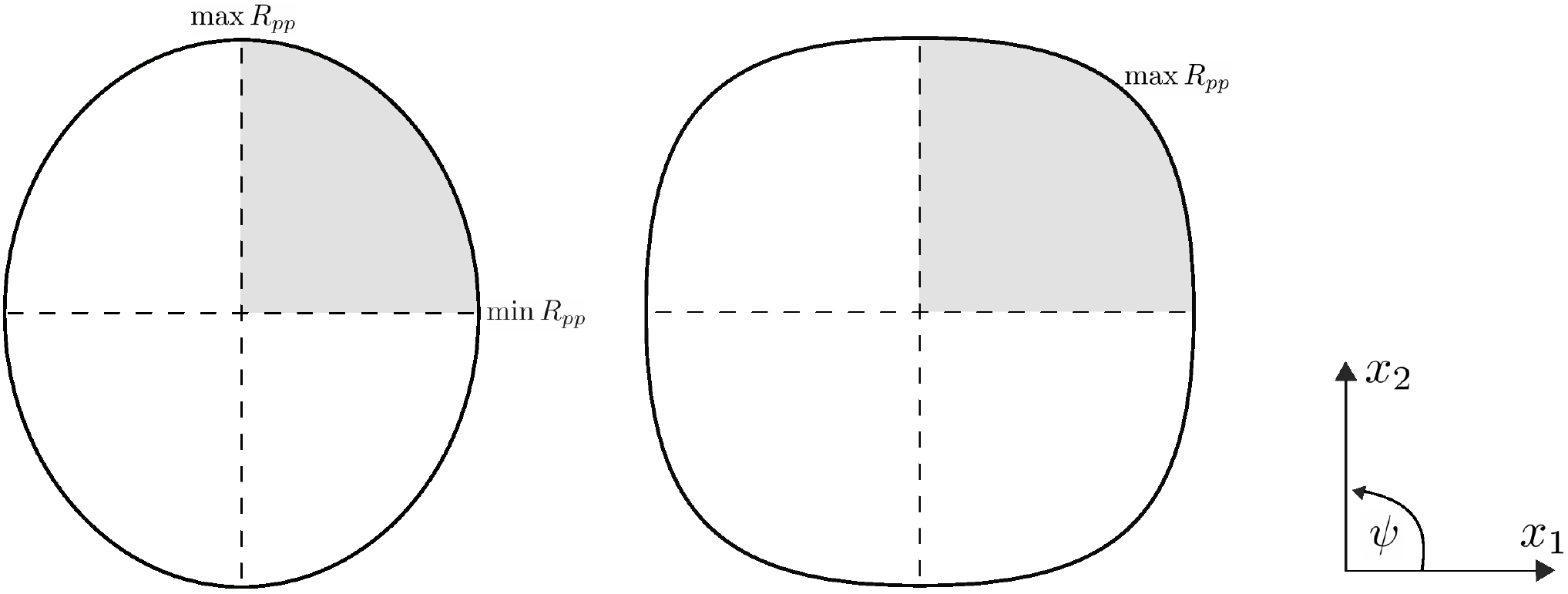}
\caption[\small{Regular and irregular azimuthal variations caused by cracks}]{\small{Two examples of azimuthal variations of amplitude caused by cracks (CAVA). Graphs have two-fold symmetry indicated by dashed lines.
The discussed quadrants are in grey. 
Graph on the left illustrates a regular shape, where $\min R_{pp}$ and $\max R_{pp}$ correspond to $\psi=0^\circ$ and $\psi=90^\circ$\,, respectively. The graph on the right presents an irregular shape, where $\max R_{pp}$ is obtained for other azimuth (in this case $\psi=45^\circ$). The coordinate system used herein is the same for all figures in the paper.  }}
\label{fig:shapes}
\end{figure} 
\subsection{Regular CAVA}\label{sec:reg}
Based on numerous works (for instance, summarised in~\citet{AVO}), we expect that cracks affect the amplitude in a most significant manner---meaning that $R_{pp}$ reaches its extreme values---while the wave propagates perpendicular or parallel to them.
Therefore, herein, we focus on the regular CAVA shape.
We propose to consider
\begin{equation}
\Delta R_{pp}:=R_{pp}(\theta,\psi_1)-R_{pp}(\theta,\psi_2)\,,
\end{equation}
where the azimuthal angles are $\psi_1=0^\circ$ and $\psi_2=90^\circ$\,.
We are particularly interested in the sign of $\Delta R_{pp}$\,. 
For instance, negative $\Delta R_{pp}$ means that amplitude has a larger value for a wave propagating parallel to cracks. 
We want to examine how much the sign is influenced by the concentration of cracks, incidence angle, or particular stiffnesses. 
Upon calculation, where we use the linearity of the difference $\Delta$\,(from expression~(\ref{vav})), we get 
\begin{equation}
\Delta R_{pp}=\frac{1}{2}\sin^2\theta\bigg\{\Delta\left[\frac{2\left(C_{44}-C_{55}\right)+\left(C_{13}-C_{23}\right)+\frac{1}{2}\tan^2\theta\left(C_{11}-C_{22}\right)}{C_{33}}\right]\bigg\}
\end{equation}
expressed in terms of effective elasticity parameters.
Without loss of generality, we choose the single set of cracks to be embedded in the lower halfspace so that $\Delta$ can be removed. 
We obtain
\begin{equation}\label{DRpp}
\Delta R_{pp}^\ell=\frac{1}{2C_{33}^\ell}\sin^2\theta\left[2\left(C_{44}^\ell-C_{55}^\ell\right)+\left(C_{13}^\ell-C_{23}^\ell\right)+\frac{1}{2}\tan^2\theta\left(C_{11}^\ell-C_{22}^\ell\right)\right].
\end{equation}
where $C_{ij}^\ell$ are the effective stiffnesses of the lower halfspace.
For small incidence angles, the third term in the numerator can be neglected.
Note that if we choose the single set of cracks to be embedded in the upper halfspace $\left(\Delta R_{pp}^u\right)$\,, the minus sign appears. 
Hence, all the conclusions regarding the behaviour of $\Delta R_{pp}^\ell$ have exactly the opposite meaning for $\Delta R_{pp}^u$\,.
Taking this into account, let us focus on $\Delta R_{pp}^\ell$ only.
Due to the weakening effect of embedded cracks (see expression~(\ref{crack:end})) and following the stability conditions, we require $C_{55}^\ell<C_{44}^\ell$\,, $\left|C_{13}^\ell\right|<\left|C_{23}^\ell\right|\,$, and $C_{11}^\ell<C_{22}^\ell$\,.
Thus, $\Delta R_{pp}^\ell$ can be either positive or negative, no matter the stiffness of the upper halfspace.
\newpage
It might be important to notice that, as opposed to phase velocities, one should not expect the magnitude of $R_{pp}$ to be larger for quasi P-waves propagating parallel to cracks. 
For instance, as shown by~\citet{Adamus}, the analogous difference between squared-velocities propagating parallel and perpendicular to cracks, for small incidence angles, is
\begin{equation}
s^2=V^2_P(\theta,\psi_1)-V^2_P(\theta,\psi_2)\approx \frac{k}{2}\sin^2\theta\left[2\left(C_{55}-C_{44}\right)+\left(C_{13}-C_{23}\right)\right]\,,
\end{equation}
where $k>0$ is a scaling factor that depends on $C_{33}$\,. 
The above expression is similar to~(\ref{DRpp})---where the third term responsible for large angles would be neglected---however, difference $C_{44}-C_{55}$ has the opposite sign.
Assuming non anomalous case of $C_{13},C_{23}>0$\,, we expect $s^2$ to be negative.
It means that velocity should be larger for a wave propagating parallel to cracks.
This is not the case for $R_{pp}\,$.

Let us express $\Delta R^\ell_{pp}$ in terms of VTI background elasticities of a lower halfspace (with a certain abuse of notation, denoted same as background stiffnesses of an arbitrary halfspace) and crack compliances $Z_N$ and $Z_T$\,.
We get
\begin{equation}
\begin{aligned}
\Delta R_{pp}^\ell=\sin^2\theta&\left[\frac{1+Z_Nc_{11}}{c_{33}+Z_N\left(c_{11}c_{33}-c_{13}^2\right)}\right]\\
&\,\,\,\,\qquad \left[\frac{c_{44}^2Z_T}{1+Z_Tc_{44}}-\frac{2c_{13}c_{66}Z_N}{1+Z_Nc_{11}}-\tan^2\theta\frac{\left(c_{11}c_{66}-c_{66}^2\right)Z_N}{1+Z_Nc_{11}}\right]
\end{aligned}
\end{equation}
that can be reduced to
\begin{equation}\label{crucial}
\Delta R_{pp}^\ell=k\sin^2\theta
\left[\chi c_{44}^2-2c_{13}c_{66}-\tan^2\theta\left(c_{11}c_{66}-c_{66}^2\right)\right]\,,
\end{equation}
where
\begin{equation}
k:=\frac{Z_N}{c_{33}+Z_N\left(c_{11}c_{33}-c_{13}^2\right)}
=\frac{8c_3e}{\left(c_1^2-c_{13}^2\right)\left(8c_3e+3\dfrac{c_{33}}{c_1}\right)}
\geq0
\end{equation}
and
\begin{equation}
\chi:=\frac{Z_T\left(1+Z_Nc_{11}\right)}{Z_N\left(1+Z_Tc_{44}\right)}=\frac{2\left(c_1^2-c_{13}^2\right)+\frac{16}{3}c_{11}c_1c_3e}{c_{44}c_1c_3\left(c_2+c_3-c_4+\frac{16}{3}e\right)}\geq0\,.
\end{equation}
Scaling factor $k$ grows with increasing concentration of cracks $e$ (assuming that stiffnesses are fixed), since $\partial_ek\geq0$\,.
Also, due to stability conditions, $k$ must be positive.
Since $k$ cannot change its sign, more essential---in the context of the shape of azimuthal variations---is the content of the squared brackets in expression~(\ref{crucial}). 
Given stability conditions, each of the three terms in squared brackets must be positive. 
However, minus signs in front of the second and third terms make it difficult to anticipate the sign of $\Delta R_{pp}^\ell$\,.
For instance, the effect of growing $e$ on $\chi$ is not apparent. 
Function $\chi(e)$ may have extrema; only specific relations among stiffnesses assure that this function is monotonic.
If 
\begin{equation}\label{ineq}
c_{11}c_1c_3\left(c_2+c_3-c_4\right)\geq2\left(c_1^2-c_{13}^2\right)\,
\end{equation}
then $\partial_e\chi\geq0$\,.
Inequality~(\ref{ineq}) is satisfied if $c_{11}\geq c_{33}\wedge c_{11}\geq2c_{44}$\,, which is a reasonable condition, since in VTI media, horizontal velocity is usually greater than the vertical one and, in general, horizontal P/S velocity ratio is greater than $\sqrt{2}\,$.
We have performed Monte Carlo simulations, where a million examples of VTI backgrounds satisfying stability conditions were chosen. The values of stiffnesses were distributed uniformly, and their ranges were selected based on the minimum and maximum values measured by~\citet{Wang}.
In $93.81\%$ of cases, inequality~(\ref{ineq}) was satisfied.
Thus, in a great majority of cases, $\partial_e\chi\geq0$ is true.
The influence of each stiffness on $\chi$ is even more complicated; again, functions may have many extrema, and lengthy inequalities must be satisfied to render them monotonic. 
On the other hand, the influence of stiffnesses on the second term is trivial.
It grows for increasing $c_{13}$ or $c_{66}\,$. 
The third term grows with increasing incidence angle and $c_{11}$\,, but dependence on $c_{66}$ is not obvious.
To remove the ambiguity associated with $c_{66}$ and to get more insight into expression~(\ref{crucial}), we propose to focus on proportions between elasticity parameters, namely, $p_{11}=c_{11}/c_{66}$\,, $p_{33}=c_{33}/c_{66}$\,, $p_{13}=c_{13}/c_{66}$\,, and $p_{44}=c_{44}/c_{66}$\,. 
We obtain
\begin{equation}\label{xx}
\Delta R_{pp}^\ell=k_p\sin^2\theta\,
\left[\chi p_{44}^2-2p_{13}-(p_{11}-1)\tan^2\theta\right]=:k_p\sin^2\theta\,
\left(a-b\tan^2\theta\right)\,,
\end{equation}
where $k_p:=c_{66}^2k$ depends on $p_{11}$\,, $p_{33}$\,, $p_{13}$\,, and $p_{44}$ solely.
Also, $\chi$ can be expressed in terms of $e$ and the aforementioned proportions only.
Having this simplified form, we expect negative $\Delta R^\ell_{pp}$ for smaller $e$ and $p_{44}$\,, significant values of $p_{11}$ and $p_{13}$\,, and larger incidence angle.
On the other hand, we expect positive $\Delta R^\ell_{pp}$ for larger $e$ and $p_{44}$\,, smaller values of $p_{11}$ and $p_{13}$\,, and smaller $\theta$\,.  
An interesting case might happen when $e$ and $\theta$ are neither very small nor very large, then $p_{11}$ and $p_{13}$ may have a deciding influence on the sign of $\Delta R_{pp}^\ell$\,.
In such a situation, if rocks are gas-bearing---where $p_{11}$ and $p_{13}$ should be small---we can expect that the reflection coefficient will be bigger for a wave propagating perpendicular to cracks (positive $\Delta R^\ell_{pp}$). 
If there is no gas, we expect negative $\Delta R_{pp}^\ell$\,.
To sum up, we conjecture that for
 \begin{itemize}
 \item{small $e$ and large $\theta\,$: expect $\Delta R_{pp}^\ell<0$ and $\Delta R_{pp}^u>0$\,,} 
 \item{moderate $e$ and $\theta$\,, and brine-saturated or dry rocks: expect $\Delta R_{pp}^\ell<0$ and $\Delta R_{pp}^u>0$\,,}
 \item{moderate $e$ and $\theta$\,, and gas-saturated rocks: expect $\Delta R_{pp}^\ell>0$ and $\Delta R_{pp}^u<0$\,,}
 \item{large $e$ and small $\theta$\,: expect $\Delta R_{pp}^\ell>0$ and $\Delta R_{pp}^u<0$\,.} 
\end{itemize}

As we have already discussed, the shape of azimuthal variations of a VTI medium with an embedded set of aligned cracks (with a normal parallel to the $x_1$-axis) essentially depends on six factors: $e$\,, $\theta$\,, $p_{11}$\,, $p_{33}$\,, $p_{13}$\,, and $p_{44}$\,.
Due to complicated forms of $Z_N$ and $Z_T$\,, it is hard to grasp each factor's exact contributions to $\Delta R_{pp}^\ell$\,. 
Therefore, we propose to consider a simpler situation of an isotropic background.
In such a case, the number of independent shape factors reduces to three: $e$\,, $\theta$\,, and $p=c_{11}/c_{44}>4/3$\,.
Factor $\chi$ can be written as
\begin{equation}
\chi=\frac{16c_{11}^2e+12c_{11}c_{44}-12c_{44}^2}{16c_{11}c_{44}e+9c_{11}c_{44}-6c_{44}^2}=
\frac{16p^2e+12p-12}{16pe+9p-6}>0\,,
\end{equation}
so we obtain
\begin{align}\label{crucial2}
\Delta R_{pp}^\ell&=\frac{4e}{c_{44}(c_{11}-c_{44})(16e+3)}\sin^2\theta\,
\left[\chi c_{44}^2-2c_{44}(c_{11}-2c_{44})-\tan^2\theta\left(c_{11}c_{44}-c_{44}^2\right)\right] \nonumber \\
&=\frac{4e}{(p-1)(16e+3)}\sin^2\theta\,
\left[
\frac{-p^2(16e+18)+p(64e+60)-36}{p(16e+9)-6}
-(p-1)\tan^2\theta
\right] \nonumber \\
&=:k_p^{\rm{iso}}\sin^2\theta\,\left[a^{\rm{iso}}-b^{\rm{iso}}\tan^2\theta\right]
\,.
\end{align}
Let us discuss the influences of sole $e$ and sole $p$ on scaling factor $k_p^{\rm{iso}}$ and two terms in squared brackets; $a^{\rm{iso}}$ and $b^{\rm{iso}}$\,.
Since $\partial_ek_p^{\rm{iso}}>0$ and $\partial_ea^{\rm{iso}}>0$\,, we know that $k_p^{\rm{iso}}$ and $a^{\rm{iso}}$ increase with growing $e\,$.
On the other hand, $\partial_pk_p^{\rm{iso}}<0$\,, $\partial_pa^{\rm{iso}}<0$\,, and $\partial_pb^{\rm{iso}}>0$\,; thus, $k_p^{\rm{iso}}$ and $a^{\rm{iso}}$ decrease, but $b^{\rm{iso}}$ increase with growing $p$\,.
Now, let us discuss influences of sole $e$ and sole $p$ on $\Delta R_{pp}^\ell$\,.
To do so, we assume a fixed incidence angle.
Due to increasing $e$\,, we expect
$\Delta R_{pp}^\ell$ to grow
(always true if $\Delta R^\ell_{pp}>0$ for $\forall e$).  
In case $p$ increases, we anticipate
$\left|\Delta R_{pp}^\ell\right|$ to diminish
(always true if $\Delta R^\ell_{pp}>0$ for $\forall e$).  
The last variable that contributes to the reflection coefficient is the incidence angle. 
Since $\partial_\theta (b^{\rm{iso}}\tan^2\theta)>0\,$, growing $\theta$ renders $\Delta R_{pp}^\ell$ more likely to be negative.
Considering the above analysis, we expect negative $\Delta R_{pp}^\ell$ for small $e$\,, large $p$\,, and large $\theta$\,. 
On the other hand, we expect positive $\Delta R_{pp}^\ell$ for large $e\,$, small $p\,$, and small $\theta\,$. 
If the concentration of cracks and incidence angle are neither very small nor large, we anticipate a significant role of rock's saturation. 
Therefore, for the isotropic background, we conjecture the same bullet points as for the VTI  surroundings.
They appear to be more convincing for the isotropic case, where fewer unknowns are involved, and $\chi$ is simplified.
\subsection{Irregular CAVA}\label{sec:irr}
In this section, we examine the irregular case of CAVA.
In other words, we check if and when $R_{pp}$  presents extreme values not only for angles normal or parallel to cracks.
Mathematically, the irregularity occurs if and only if
\begin{equation}\label{condd}
\begin{gathered}
\left[
R_{pp}(\theta,\psi_{irr})>R_{pp}(\theta,0^\circ)\land
R_{pp}(\theta,\psi_{irr})>R_{pp}(\theta,90^\circ)
\right]
\\
\lor
\\
\left[
R_{pp}(\theta,\psi_{irr})<R_{pp}(\theta,0^\circ)\land
R_{pp}(\theta,\psi_{irr})<R_{pp}(\theta,90^\circ)
\right]\,,
\end{gathered}
\end{equation}
where $\psi_{irr}\in(0^\circ,90^\circ)$\,.
The above condition is true for either lower or upper halfspace. 
Again, without loss of generality, we assume that cracks are embedded in the lower halfspace and consider differences between reflection coefficients.
We define
\begin{equation}\label{irr}
\Delta R_{pp\psi}:=R_{pp}(\theta,0^\circ)-R_{pp}(\theta,\psi_{irr})\,
\end{equation}
so that condition~(\ref{condd}) can be simply formulated as
\begin{equation}\label{condd2}
\left[
\Delta R_{pp\psi}^\ell<0
\land
\Delta R_{pp\psi}^\ell<\Delta R_{pp}^\ell
\right]
\lor
\left[
\Delta R_{pp\psi}^\ell>0
\land
\Delta R_{pp\psi}^\ell>\Delta R_{pp}^\ell
\right]\,.
\end{equation}
Hence, to examine the irregularity, first we need to derive $\Delta R^\ell_{pp\psi}$\,.
Upon algebraic operations, expression~(\ref{irr}) can be written as
\begin{equation}\label{irr2}
\begin{aligned}
\Delta R_{pp\psi}^\ell&=\frac{1}{2C_{33}^\ell}\sin^2\theta\sin^2\psi
\bigg\{2C_{44}^\ell-2C_{55}^\ell+C_{13}^\ell-C_{23}^\ell\\
&+\frac{1}{2}\tan^2\theta\left[2C_{11}^\ell-2C_{12}^\ell-4C_{66}^\ell+\sin^2\psi\left(2C_{12}^\ell+4C_{66}^\ell-C_{11}^\ell-C_{22}^\ell\right)\right]\bigg\}\,,
\end{aligned}
\end{equation}
where $C_{ij}^\ell$ are the effective stiffnesses of a lower halfspace. 
It has a similar form to expression~(\ref{DRpp}).
Analogously to the previous section, we express $\Delta R_{pp\psi}^\ell$ in terms of proportions between background elasticities, $p_{ij}$.
We obtain
\begin{align}\label{irxx}
\Delta R_{pp\psi}^\ell&=k_p\sin^2\theta\sin^2\psi
\bigg\{\chi p_{44}^2-2p_{13}-\tan^2\theta \nonumber \\
&\qquad \left[\frac{p_{11}Z_N-Z_T-\sin^2\psi(Z_N-Z_T+Z_NZ_Tc_{66}(1-p_{11}))}{Z_N(1+c_{66}Z_T)}\right]\bigg\}
\nonumber \\
&=:k_p\sin^2\theta\sin^2\psi
\left(a-\beta\tan^2\theta\right)\,.
\end{align}
If we express $Z_N$ and $Z_T$ in terms of proportions $p_{11}$\,, $p_{33}$\,, $p_{13}$\,, and $p_{44}$\,, then $c_{66}$ do cancel.
Expression~(\ref{irxx}) is analogous to expression~(\ref{xx}); similarly to $\Delta R_{pp}^\ell$\,, coefficient $\Delta R_{pp\psi}^\ell$ can be either negative or positive. 
The relation between $\Delta R_{pp\psi}^\ell$ and $\Delta R_{pp}^\ell$ is not obvious.
In general, $\sin^2\psi<1$ in front of the curly brackets decreases $\left|\Delta R_{pp\psi}^\ell\right|$\,; this trigonometric function in the expression of $\left|\Delta R_{pp}^\ell\right|$ equals one.
On the other hand, the contribution of $\beta$ inside $\left|\Delta R_{pp\psi}^\ell\right|$ can be either smaller or larger as compared to analogous contribution of $b$ inside $\left|\Delta R_{pp}^\ell\right|$\,; note that
\begin{equation}\label{bminb}
\beta-b=\frac{(Z_N-Z_T)\left(1-\sin^2\psi\right)-Z_NZ_Tc_{66}\left(1-\sin^2\psi\right)\left(p_{11}-1\right)}{Z_N\left(1+c_{66}Z_T\right)}
\end{equation}
can be either positive or negative that is governed by the azimuth, crack concentration, and proportions $p_{ij}$\,.
Hence, both irregular CAVA, namely, $\Delta R_{pp\psi}^\ell<0\land\Delta R_{pp\psi}^\ell<\Delta R_{pp}^\ell$ or $\Delta R_{pp\psi}^\ell>0\land \Delta R_{pp\psi}^\ell>\Delta R_{pp}^\ell$ seem to be possible.
However, we can show the case where irregularity is impossible for all azimuths by introducing certain assumptions.
Assume that $Z_N\leq Z_T$\,, $p_{11}>1$\,, and consider $\Delta R_{pp\psi}^\ell<0$\,.
Then, relation $\Delta R_{pp\psi}^\ell-\Delta R_{pp}^\ell$\,, namely,
\begin{equation}\label{case1}
k_p\sin^2\theta\left[\sin^2\psi\left(a-\beta\tan^2\theta\right)-\left(a-b\tan^2\theta\right)\right]
=:
k_p\sin^2\theta\left(x\sin^2\psi-y\right)\,
\end{equation}
must be positive, since we easily obtain $\beta<b$ that leads to $x>y$\,, where $y<0$\,.
Following~\citet{GrechkaKachanov}, $Z_N>Z_T$ is true for isotropic rocks having negative Poisson's ratio, which is a rare case. 
Also, $p_{11}>1$ is a typical situation corresponding to horizontal P-wave faster than S-wave.
Hence, we can state that the irregularity is unlikely to occur for rocks with regular Poisson's ratio, where $\Delta R_{pp\psi}^\ell<0$ (or $\Delta R_{pp\psi}^u>0$)\,.

Let us perform analogous analysis, assuming isotropic, not VTI, background.
We obtain
\begin{equation}\label{relations}
\Delta R_{pp\psi}^\ell=k_p^{\rm{iso}}\sin^2\theta\sin^2\psi\left( a^{\rm{iso}}-\beta^{\rm{iso}}\tan^2\theta\right)\,,
\end{equation}
where 
\begin{equation}
\beta^{\rm{iso}}=\frac{12p\left(\frac{3}{4}p-1\right)+(2-p)\left(6-3\sin^2\psi\right)+16pe(p-1)\sin^2\psi}{16pe+9p-6}>0\,.
\end{equation}
For isotropy, $Z_N=Z_T$ is tantamount to $p=2$ and zero Poisson's ratio.
Again, we try to compare $\beta^{\rm{iso}}$ with $b^{\rm{iso}}$. 
We notice that $\partial_p(\beta^{\rm{iso}}-b^{\rm{iso}})<0$\,.
Also, if $p=2$, then, $\beta^{\rm{iso}}<b^{\rm{iso}}$. 
It means that $Z_N\leq Z_T$ is tantamount to $\beta^{\rm{iso}}<b^{\rm{iso}}$\,. 
On the other hand, if $Z_N>Z_T$, then $\beta^{\rm{iso}}$ can be either larger or smaller than $b^{\rm{iso}}$\,.
Assume that $Z_N\leq Z_T$ and consider $\Delta R_{pp\psi}^\ell<0$\,.
Then, relation $\Delta R_{pp\psi}^\ell-\Delta R_{pp}^\ell$\,, namely,
\begin{equation}\label{case2}
\begin{aligned}
k_p^{\rm{iso}}\sin^2\theta\left[\sin^2\psi\left(a^{\rm{iso}}-\beta^{\rm{iso}}\tan^2\theta\right)-\left(a^{\rm{iso}}-b^{\rm{iso}}\tan^2\theta\right)\right] \hphantom{xxxxxxxxxxxxxxx}\\
\qquad=k_p^{\rm{iso}}\sin^2\theta\left(x^{\rm{iso}}\sin^2\psi-y^{\rm{iso}}\right)\,
\end{aligned}
\end{equation}
must be positive, since $\beta^{\rm{iso}}<b^{\rm{iso}}$ that leads to $x^{\rm{iso}}>y^{\rm{iso}}\,$, where $y^{\rm{iso}}<0\,$.
In other words, the irregularity cannot occur for rocks with positive Poisson's ratio, where $\Delta R_{pp\psi}^\ell<0$ (or $\Delta R_{pp\psi}^u>0$)\,.
We illustrate the analysis from this section in Figure~\ref{fig:irr}, where some examples of possible irregular CAVA are discussed.

\begin{figure}[!htbp]
\includegraphics[width=0.9\textwidth]{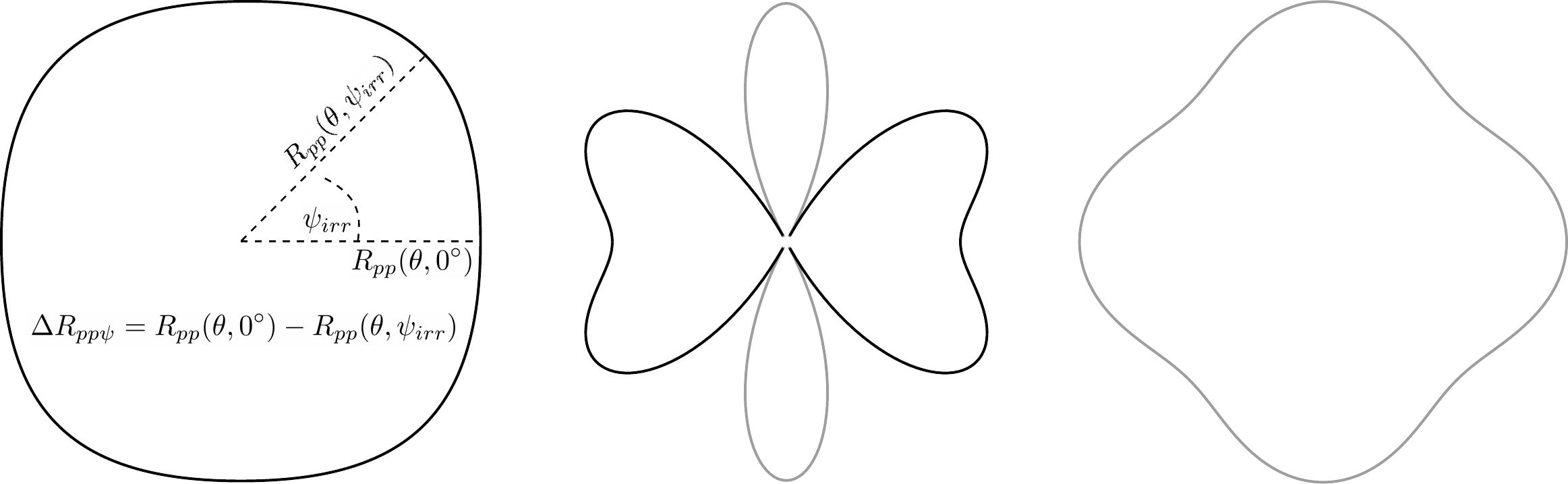}
\caption[\small{Three examples of irregular azimuthal variations}]{\small{Three examples of irregular CAVA for fixed incidence are shown. Let us choose black and grey colours to represent positive and negative reflection coefficient, respectively. Despite remarkably different shapes, the irregularity of each graph corresponds to $\Delta R_{pp\psi}<0$\,.
Based on the analysis from Section~\ref{sec:irr}, the above graphs are very unlikely to occur if cracks are embedded in a lower halfspace, since $\Delta R^\ell_{pp\psi}<0$ occurs for very low Poisson's ratio only (rarely presented in seismology). Note that if we switch colours, the irregularity corresponds to $\Delta R_{pp\psi}>0$\,. In such a case, the shapes are unlikely to occur if cracks are embedded in an upper---instead of lower---halfspace, since $\Delta R^u_{pp\psi}>0$ happens for very low Poisson's ratio only.}}
\label{fig:irr}
\end{figure} 
\subsection{CAVA reversing process}
Previously, we have examined what parameters decide whether the reflection coefficient is larger for azimuths parallel or perpendicular to cracks (bullet points in Section~\ref{sec:reg}).
Also, we have discussed that extreme values of $R_{pp}$ can correspond to other than the aforementioned directions.
In such a case, we experience irregularity.
Irregular CAVA is less or more likely to occur, depending on the sign of $\Delta R_{pp\psi}$ and $Z_T-Z_N$ (see Section~\ref{sec:irr}).  
In turn, $\Delta R_{pp\psi}$ depends on the concentration of cracks, incidence angle, and stiffnesses.
Due to complicated form of expression~(\ref{case1}) or~(\ref{case2}), it is hard to grasp what proportions of stiffnesses, or magnitudes of $e$ and $\theta$\,, lead to irregularity.
In other words, we are not able to propose the analogous bullet points as it was done in Section~\ref{sec:reg}.
However, there is a specific case when azimuthal variations are extremely likely to be irregular.
This situation occurs during ``CAVA reversing process'' that we discuss below.

Consider again expression~(\ref{xx}) and assume certain values of stiffnesses, $e$\,, and $\theta$ for which $\Delta R_{pp}^\ell<0$ (equivalently $\Delta R_{pp}^u>0$)\,.
In such a case, $R_{pp}$ is larger in parallel than in a perpendicular direction to the cracks.
CAVA may be either regular or irregular.
Based on the analysis from Section~\ref{sec:reg}, growing $e$ increases $\Delta R_{pp}^\ell$\,.
Hence, if we continuously increase $e$ (while the other parameters are fixed), we reach $\Delta R_{pp}^\ell=0$ defined as a reversing point, and subsequently, $\Delta R_{pp}^\ell>0$\,.
The above process expresses CAVA reversing, illustrated in Figure~\ref{fig:rev}. 
\begin{figure}[!htbp]
\includegraphics[width=0.8\textwidth]{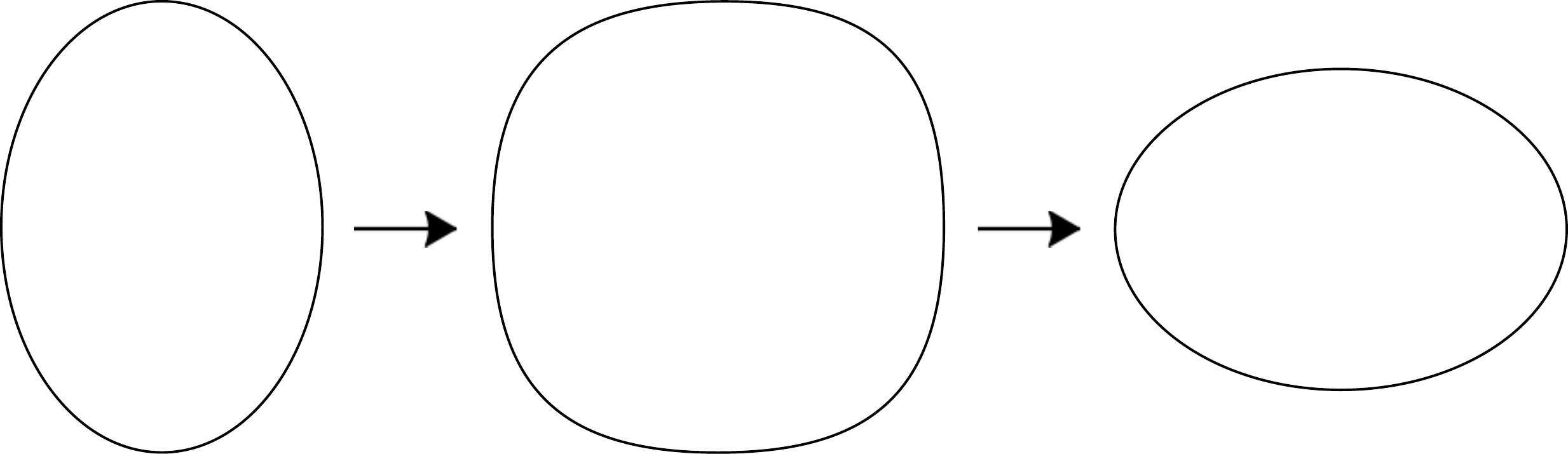}
\caption[\small{Azimuthal reversing process for growing crack concentration}]{\small{Example of CAVA reversing process for continuously growing concentration of cracks (indicated by arrows). Middle graph illustrates the reversing point, where $R_{pp}(\theta,0^\circ)=R_{pp}(\theta,90^\circ)$ equivalent to $\Delta R_{pp}=0$\,. The reversing process is possible if the first graph presents either $\Delta R_{pp}^\ell<0$ or $\Delta R_{pp}^u>0$\,. Hence, the above example illustrates either positive $R_{pp}$ calculated for a model with cracks embedded in a lower halfspace or negative $R_{pp}$ corresponding to cracks embedded in an upper halfspace. }}
\label{fig:rev}
\end{figure} 
The reversing process cannot occur if $\Delta R_{pp}^\ell>0$\,.
From the seismological perspective, CAVA at the reversing point is very likely to be irregular, as shown in the following Lemma.
\begin{lemma}
The polar graph of the PP reflection coefficient at CAVA reversing point---where $\Delta R^\ell_{pp}=0$\,, $e>0$\,, and $\theta>0$---is either irregular or a regular circle. 
\end{lemma}
\begin{proof}
Assume that CAVA is not irregular and consider a VTI background.
The graph is not irregular at reversing point if and only if $\Delta R_{pp\psi}^\ell=0$ for $\forall\psi_{irr}\in(0^\circ,90^\circ)$\,; in such a case, it represents a regular circle.
\end{proof}
We can show that a circle at the reversing point is unlikely in the seismological context. 
To have a regular shape, we require that
\begin{equation}\label{lem}
k_p\sin^2\theta\left(x\sin^2\psi-y\right)=0\,.
\end{equation}
To satisfy stability conditions, constant $k_p$ must be greater than zero (for $e>0$)\,.
Expression~(\ref{lem}) is true if $x\sin^2\psi=y$\,.
Since $x$ and $y$ for given $\theta$ are constants, CAVA is regular if and only if $x=y=0$\,, which---given definition of these constants from expression~(\ref{case1})---is true for $\beta=b$\,. In turn, considering  expression~(\ref{bminb}), we see that $\beta$ equals $b$ if and only if 
\begin{equation}
Z_N-Z_T-Z_NZ_Tc_{66}(p_{11}-1)=0
\end{equation}
that is extremely unlikely in the seismological context---since usually $p_{11}>1$ and $Z_T>Z_N$---but physically possible. Analogical analysis can be performed for an isotropic case. 
\subsection{Magnitude of CAVA}
So far, we have thoroughly discussed the shape of azimuthal variations caused by cracks.
However, we have not investigated the influence of $e$ on the magnitude of the reflection coefficient.
Due to the complexity of Vavry\v{c}uk-P\v{s}en\v{c}\'{i}k approximation, analytical analysis is challenging to perform.
Even if we assume that cracks are embedded in one of the halfspaces (as we did in previous sections), we cannot get rid of $\Delta$ in expression~(\ref{vav}), and stiffnesses of two halfspaces have to be considered.
Therefore, we propose numerical instead of analytical analysis. 
We assume that both halfspaces have a VTI background, whereas cracks are embedded in the lower one. 
We perform three MC simulations to obtain $R^\ell_{pp}(\theta,\psi,e)$ for different incidence angles; $\theta=15^\circ$\,, $\theta=30^\circ\,$, and $\theta=45^\circ$\,.
Let us discuss a single simulation.
MC chooses one--thousand elasticity tensors for each halfspace (again stiffnesses are distributed uniformly and their range is taken from~\citet{Wang}).
Then, $R^\ell_{pp}(\psi,e)$ is calculated for azimuths $\psi=0^\circ$\,, $\psi=45^\circ$\,, and $\psi=90^\circ$\,.
Finally, to understand the influence of cracks on the magnitude of the reflection coefficient, we compute a derivative of $R^\ell_{pp}(e)$ with respect to $e$\,.
We check in what percentage of cases $R^\ell_{pp}(e)$ is a monotonic function; continuously decreases/increases for growing $e$\,.
In other words, we focus on $\partial_eR^\ell_{pp}<0$ or $\partial_eR^\ell_{pp}>0$ for all $e\in[0,1]$\,.
We present our findings in Table~\ref{tab:mag}.
\begin{table}[!htbp]
\begin{tabular}
{ccccc}
\toprule
 & $\psi=0^\circ$ &  $\psi=45^\circ$ &   $\psi=90^\circ$ &  all combined  \\[+0.05cm]
 \cmidrule{1-5}
 sim. $I$ ($\theta=15^\circ$) & 73.9\,/\,7.3 & 84.3\,/\,4.1&91.4\,/\,2.3 & 73.9\,/\,2.3 \\
  sim. $II$ ($\theta=30^\circ$) & 69.0\,/\,6.1 & 80.5\,/\,4.0& 83.5\,/\,3.4& 67.7\,/\,3.1 \\
   sim. $III$ ($\theta=45^\circ$) & 76.7\,/\,3.0 & 84.5\,/\,1.9&82.4\,/\,2.1 & 72.4\,/\,1.3 \\
  \bottomrule
\end{tabular}
\caption[\small{Decreasing/increasing reflection coefficient vs. growing crack concentration}]{\small{Numbers refer to the percentages of cases, where reflection coefficient continuously decreases/increases for growing crack concentration $e\in[0,1]$\,.
To obtain the results, $\partial_eR^\ell_{pp}$ are computed for one-thousand examples of interfaces drawn three times in MC simulations (sim. $I$, sim. $II$, and sim. $III$).
In each simulation, a different incidence angle is chosen.
Percentages are presented for particular azimuthal angles (columns 2--4).
The last column indicates a decrease/increase of $R^\ell_{pp}$ that must occur for all azimuths combined, namely, $\psi=0^\circ$\,, $\psi=45^\circ$\,, and $\psi=90^\circ$\,.
VTI backgrounds with cracks (with a normal parallel to the $x_1$-axis) embedded in the lower halfspace are assumed.  
}}
\label{tab:mag}
\end{table}
Based on the results from the last column, we notice that for most of the chosen models, $R^\ell_{pp}(e)$ decreases with the growing concentration of cracks.
Such a decrease is more probable for $\theta=15^\circ$\,; however, in this context, the incidence angle's influence is not significant.
In general, an increase of $R^\ell_{pp}(e)$ seems to be very unlikely. 
Columns two to four---apart from giving us insight into magnitudes---provide us with interesting information on CAVA shape. 
We notice that the reflection coefficient is more likely to decrease in the direction parallel to cracks than the perpendicular one.
Also, in some cases we can expect irregularity since there exist examples, where $R^\ell_{pp}(e)$ decreases for $\psi=45^\circ$\,, but not for $\psi=0^\circ$ and $\psi=90^\circ$\,.
Such irregularity is more likely to occur for large incidence angles.

Note that if cracks are embedded in the upper halfspace, the results are exactly the opposite. CAVA is likely to increase but unlikely to decrease. 
To sum up, in at least two--third scenarios, growing $e$ leads to a continuous decrease/increase of $R_{pp}(e)$\,, where the lower/upper medium is cracked, respectively. 
Similar results can be obtained for isotropic, instead of VTI, backgrounds.
\subsection{CAVA patterns}
In this section, we fix an incidence angle and propose patterns that consist of a series of two-dimensional polar graphs illustrating how CAVA may change with increasing crack concentration. 
To do so, we make use of the analysis of azimuthal shapes and magnitude, performed in the previous sections.
Let us enumerate essential findings and conjectures. We assume either VTI or isotropic backgrounds and cracks embedded in the lower halfspace (conclusions are the opposite if cracks are situated in the upper halfspace). 
With growing crack concentration:
\begin{enumerate}
\item{if $\theta$ is large, negative $\Delta R_{pp}$ becomes positive (reversing process),}
\item{if $\theta$ is small (or moderate, but rocks are saturated by gas), positive $\Delta R_{pp}$ remains positive,}
\item{irregularity may occur if $\Delta R_{pp\psi}>0$ (unless Poisson's ratio is very low),}
\item{irregularity usually occurs during reversing process (when negative $\Delta R_{pp}$ becomes positive),} 
\item{reflection coefficient usually decreases.}
\end{enumerate}
Based on the above findings, we propose patterns illustrated in Figures~\ref{fig:pattern1} and~\ref{fig:pattern2}.
Let us discuss Figure~\ref{fig:pattern1}.
The patterns shown therein, correspond to either $R_{pp}<0$ for $e=0$ and cracks in the lower halfspace, or $R_{pp}>0$ for $e=0$ and cracks in the upper halfspace.
The patterns are short since CAVA does not change the sign; negative $R_{pp}^\ell$ continuously decreases, whereas positive $R_{pp}^u$ continuously increases. 
Figure~\ref{fig:pattern1a} \,describes the situation, where the incidence angle is small (or moderate, but rocks are gas-bearing).
On the other hand, Figure~\ref{fig:pattern1b} refers to the case of a large $\theta$\,.
More extended patterns are shown in Figure~\ref{fig:pattern2}, where CAVA changes the sign.
Thus, they describe either $R_{pp}>0$ for $e=0$ and cracks in the lower halfspace, or $R_{pp}<0$ for $e=0$ and cracks in the upper halfspace. 
Figures~\ref{fig:pattern2a} and~\ref{fig:pattern2b} refer to small and large incidence, respectively.
\begin{figure}[!htbp]
\begin{subfigure}{.99\textwidth}
  \centering
  \includegraphics[scale=0.35]{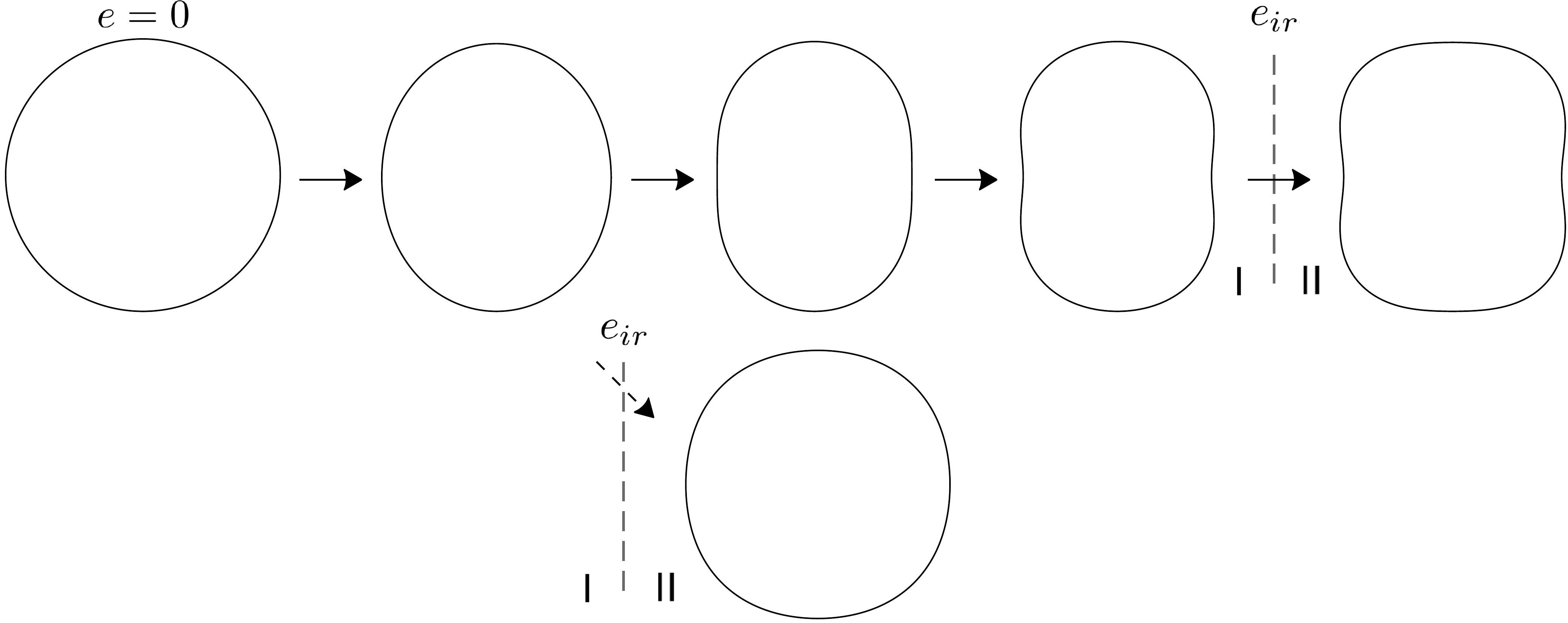}
  \caption{\footnotesize{smaller $\theta$ (gas saturation more likely)}}
  \label{fig:pattern1a}
\end{subfigure}%
\qquad
\begin{subfigure}{.99\textwidth}
\centering
   \includegraphics[scale=0.35]{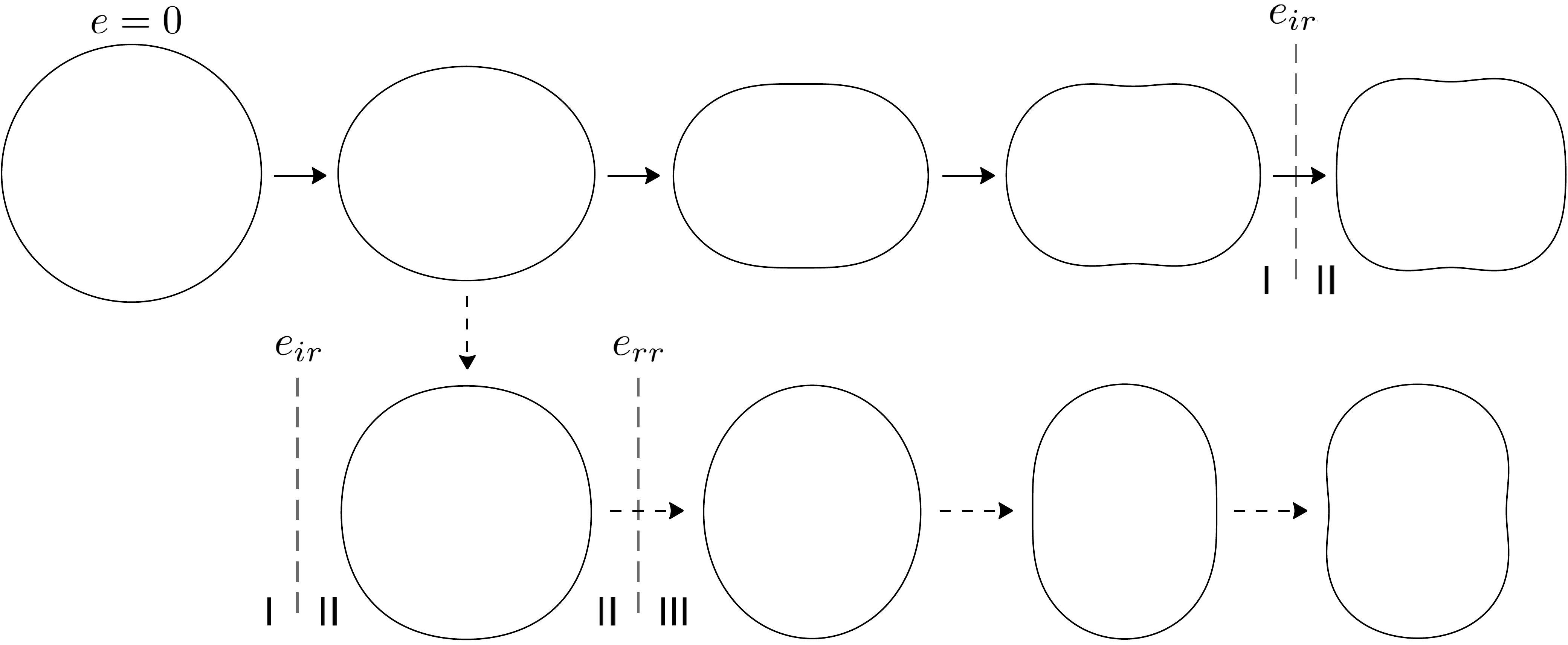}
  \caption{\footnotesize{larger $\theta$ (brine saturation more likely)}}
    \label{fig:pattern1b}
\end{subfigure}%
\caption[\small{Azimuthal variations vs. increasing crack concentration: Short patterns}]{\footnotesize{
Short CAVA patterns illustrate changes of azimuthal variations of reflection coefficient with increasing concentration of cracks, $e$ (indicated by arrows).
Changes in shape---not in magnitude---are reflected only.
Either VTI or isotropic backgrounds are assumed. Cracks are embedded either in a lower ($\ell$) or upper ($u$) halfspace and have a normal parallel to the $x_1$-axis. 
One pattern corresponds to a small incidence angle (or moderate angle and gas-bearing rocks), whereas the second pattern refers to large $\theta$ (or moderate angle and brine-bearing rocks).
Both schemes illustrate either decreasing, negative $R_{pp}^\ell$ or increasing, positive $R_{pp}^u$.
Phases $I$, $II$, and $III$ consist of regular, irregular, and reversed CAVA, respectively.
Boundary between phases $I$ and $II$ is denoted by concentration $e_{ir}$\,. Boundary between phases $II$ and $III$ is described by $e_{rr}$\,.
Dashed arrows indicate some possible ``shortened'' patterns, where the irregular phase occurs at an earlier stage.
Additionally, other shortened patterns---where latter stages are absent (due to, for instance, lack of irregular and reversed phases)---are possible, but are not shown herein.
 }}
\label{fig:pattern1}
\end{figure}
\begin{figure}[!htbp]
\centering
\begin{subfigure}{.99\textwidth}
  \centering
  \includegraphics[width=0.89\textwidth]{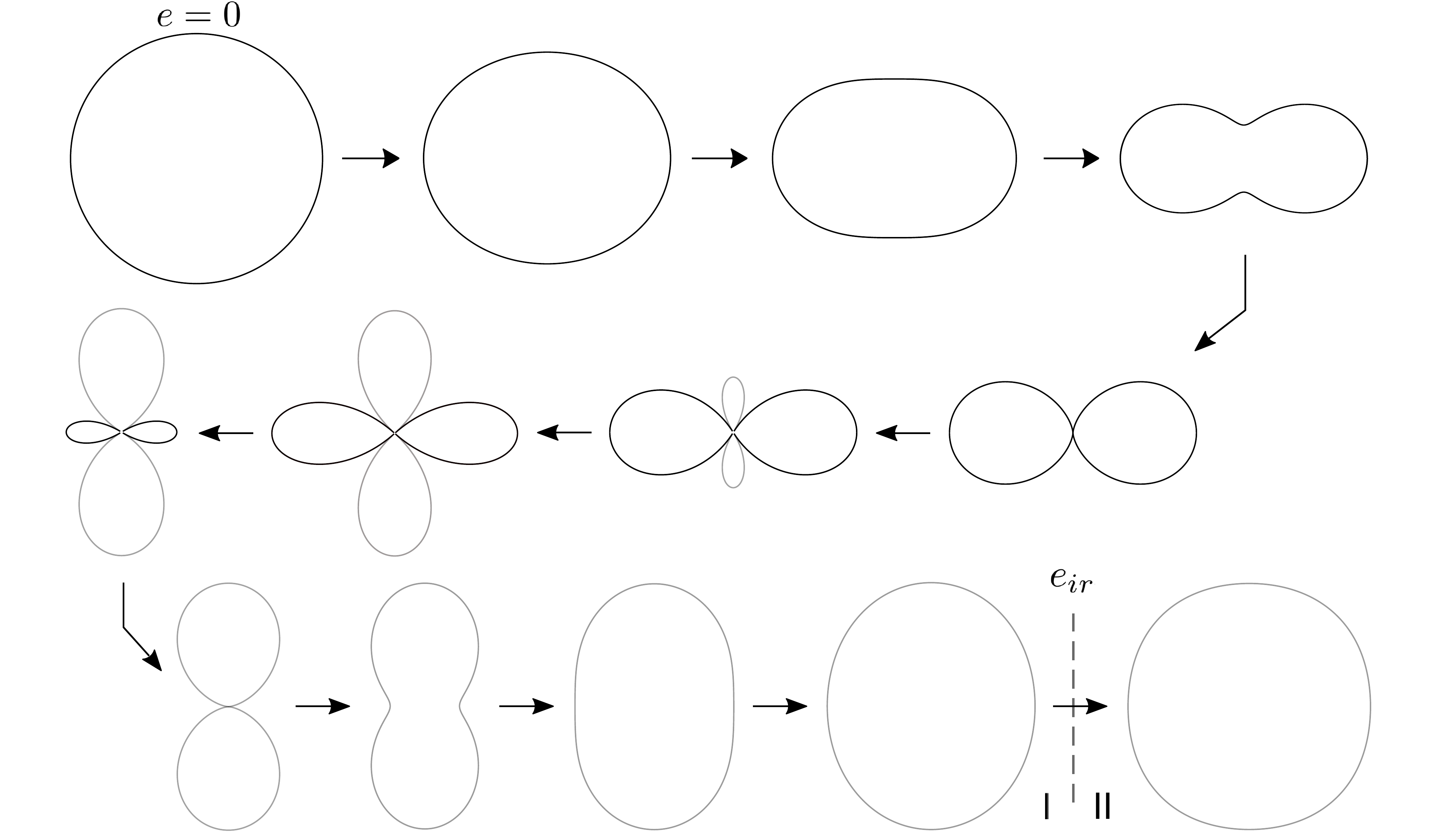}
  \caption{\footnotesize{smaller $\theta$ (gas saturation more likely)}}
  \label{fig:pattern2a}
\end{subfigure}%
\qquad
\begin{subfigure}{.99\textwidth}
  \centering
   \includegraphics[width=0.89\textwidth]{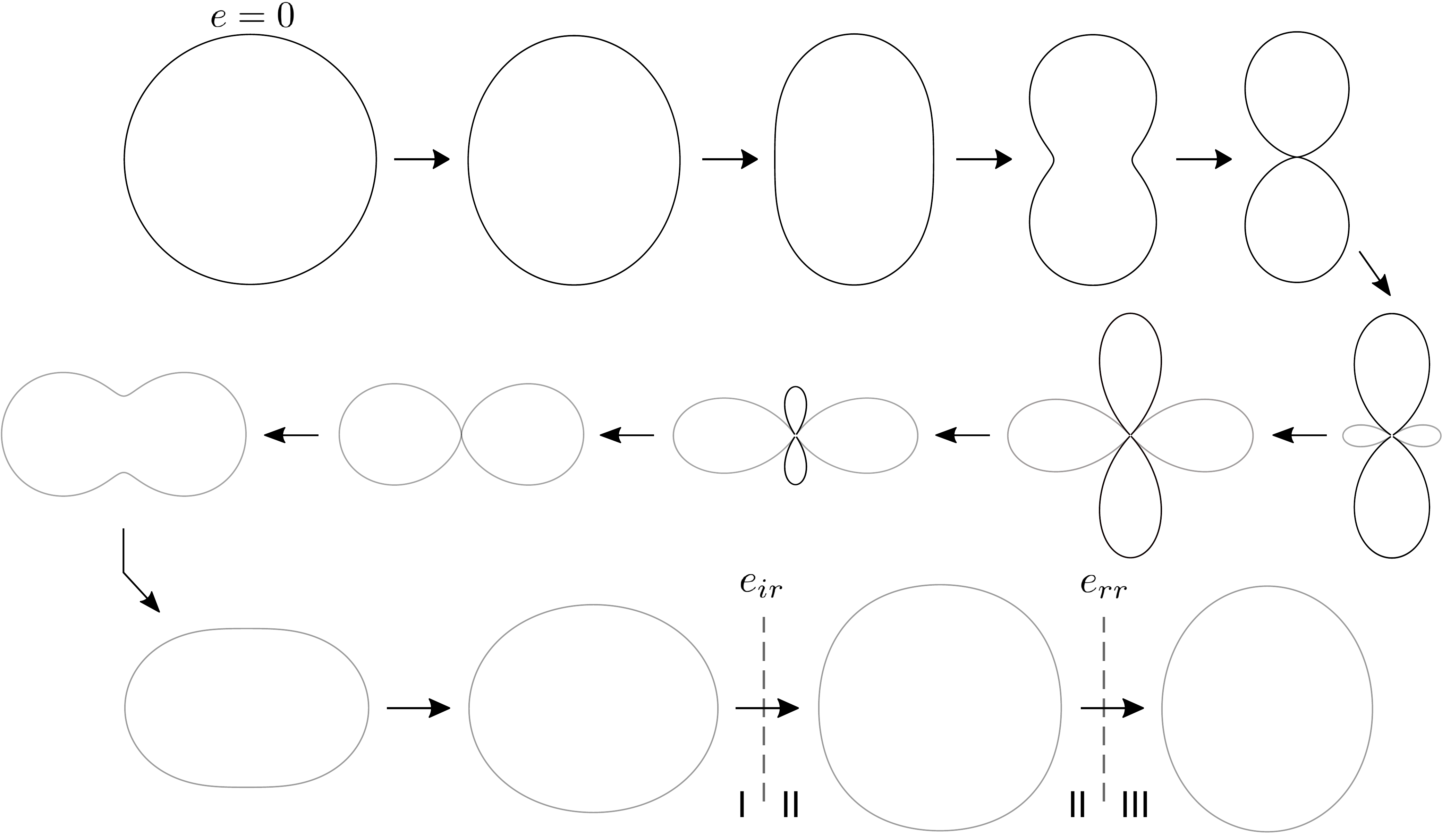}
  \caption{\footnotesize{larger $\theta$ (brine saturation more likely)}}
    \label{fig:pattern2b}
\end{subfigure}%
\caption[\small{Azimuthal variations vs. increasing crack concentration: Long patterns }]{\footnotesize{Long CAVA patterns illustrate changes of azimuthal variations of reflection coefficient with increasing concentration of cracks, $e$ (indicated by arrows).
Changes in shape and sign (different colours)---not in magnitude---are reflected only.
Either VTI or isotropic backgrounds are assumed. Cracks are embedded either in a lower ($\ell$) or upper ($u$) halfspace and have normal towards the $x_1$-axis. 
One pattern corresponds to a small incidence angle (or moderate angle and gas-bearing rocks), whereas the second pattern refers to large $\theta$ (or moderate angle and brine-bearing rocks).
Both schemes illustrate either decreasing, initially positive $R_{pp}^\ell$ or increasing, initially negative $R_{pp}^u$\,.
Phases I, II, and III consist of regular, irregular, and reversed CAVA, respectively.
Boundary between phases $I$ and $II$ is denoted by concentration $e_{ir}$\,. Boundary between phases $II$ and $III$ is described by $e_{rr}$\,.
The absence of latter stages and/or the presence of an early irregularity is possible; it leads to shortened patterns not shown herein.  }}
\label{fig:pattern2}
\end{figure}

Our CAVA patterns consist of regular, irregular, and---if a change of $\Delta R_{pp}$ sign is possible---reversed phase.
However, we have introduced some simplifications to the patterns.
First, the irregularity for some $e_{ir}$ (and reversed shape for some $e_{rr}$) may occur at earlier stage, as indicated by two distinct pattern's branches in both Figures~\ref{fig:pattern1a} and~\ref{fig:pattern1b}.
Such a possibility has not been shown in Figure~\ref{fig:pattern2}.
Therein, the irregular and reverse phases occur at the latest possible stage.
Second, not in every seismological case, the irregular or reversed azimuthal variations must occur.
In other words, the pattern may end earlier and not be full.\linebreak
The aforementioned two circumstances lead to ``shortened CAVA patterns'' linked to Figures~\ref{fig:pattern1} and~\ref{fig:pattern2}, but not shown explicitly there.
Hence, full patterns from Figures~\ref{fig:pattern1} and~\ref{fig:pattern2} should be treated as general, idealised schemes.
Note that due to the change of sign, regular CAVA gains specific shapes.
For instance, in Figures~\ref{fig:pattern2a} and~\ref{fig:pattern2b}, the third graph (counted from the upper-left corner) illustrates the last, limiting, only-convex shape (potato-like).
Then, due to concave parts, CAVA becomes peanut-like.
With growing $e\,$, it reaches infinity-like, knot-like, and shamrock-like shapes, respectively.
Exceptionally, the information on the aforementioned shapes is not induced by the analytical analysis performed in the previous sections. 
We have observed the shapes upon numerous simulations, discussed in Section~\ref{sec:four}.
Also, the irregularity may have different than oval shape, as exemplified in Figure~\ref{fig:irr}.

Let us discuss CAVA patterns in the context of inverse problems.
We notice an important but, in a way, disappointing issue.
The same CAVA (shapes and signs) are present in more than one pattern.
Hence, even if we know $R_{pp}$\,, cracks orientation, and the incidence angle, we cannot infer the sign of the reflection coefficient of the background medium (where $e=0$)\,.
Also, it is hard to grasp what the magnitude of the crack concentration is, or in which halfspace the cracks are situated.
Finally, perhaps most importantly, the difficulty in choosing the right pattern affects the correct inference on rock's saturation. 
To understand it better, consider the following example.
Assume that our CAVA is negative and regular with $\Delta R_{pp}>0$ (narrow and tall ellipse-like shape); the incidence is moderate, and cracks are embedded in the lower halfspace having a normal parallel to the $x_1$-axis.
If our shape belonged to pattern from Figure~\ref{fig:pattern1a} or~\ref{fig:pattern2a}, gas-saturation would be very probable.
On the other hand, if it belonged to pattern from Figure~\ref{fig:pattern1b} or~\ref{fig:pattern2b}, brine-saturation or no-saturation would be more probable.
Unfortunately, our CAVA matches all the aforementioned figures, so the inference on rock's saturation is difficult.
Nonetheless, there are examples where such inference is simple.
Again consider the same location of cracks and moderate incidence. 
Assume positive, regular CAVA with $\Delta R_{pp}<0$ (wide and short ellipse-like shape). Such CAVA matches Figure~\ref{fig:pattern2a} only.
Hence, in this example, gas-saturation is very probable. 
Therefore, we expect that the conjectured patterns may be useful in gas exploration despite the aforementioned difficulties.
\section{Numerical verification}\label{sec:four}
%
In this section, we use numerical techniques to verify and enrich our previous analysis, which led to the conjectured CAVA patterns.
First, we pose the following questions to better understand the nature of azimuthal variations of amplitude for cracked media. 
Are the conjectured patterns correct? How often can we expect the shortened patterns? Do patterns from Figures~\ref{fig:pattern1a} and~\ref{fig:pattern2a} really occur more often for small incidence angles and gas rocks? What phases (regular, irregular, or reversed) are the most common for specified crack densities and incidence angles?
To answer these questions, we use twenty models of interfaces between VTI elastic backgrounds. The values of stiffnesses were measured in a laboratory by~\citet{Wang}. 
In such models, we increase crack concentration in either lower or upper halfspace to obtain forty CAVA patterns for each incidence angle. We examine seven specific incidences, where $\theta\in[1^\circ,45^\circ]$\,; hence, in total, we verify two hundred eighty patterns.

Results of numerical experiments obtained for $\theta=15^\circ$ are presented in Table~\ref{tab:15}. 
Findings for the other six incidence angles are exhibited in Appendix~\ref{ap:ch8_one}.
Additionally, a \MATLAB code used to obtain the results is shown in Appendix~\ref{ap:ch8_two}.
Herein, as an example, we analyse the aforementioned table only.
Twenty models of interfaces consist of diverse geological scenarios (first column).
We examine a variety of elastic backgrounds that correspond to sedimentary rocks; sands, shales, coals, limestones, and dolomites. 
The halfspaces may be either gas or brine saturated.
Each background has a symbol assigned (second column) so that its stiffnesses can be extracted from~\citet{Wang} directly (we took values for the lowest overburden pressure).
The third and fourth columns give us important information on $R_{pp}$\,, so we infer what CAVA pattern should be expected (fifth column).
We increase crack concentration in either halfspace and obtain the actual pattern (sixth column) along with crack densities that correspond to phase boundaries (two last columns). 

To gain more insight into Table~\ref{tab:15}, consider model number one with cracks embedded in the lower halfspace.
The reflection coefficient is negative and decreases with growing crack concentration. 
We can expect either pattern from Figure~\ref{fig:pattern1}, since the incidence is neither very small nor large.
Upon continuous increase of $e$\,, CAVA reaches irregular phase at $e_{irr}=0.16$ and reversed phase at $e_{rr}=0.20$\,. 
Despite the occurrence of all phases, the pattern is shortened since the last two shapes indicated by dashed arrows in Figure~\ref{fig:pattern1b} do not appear.
This time, consider model number two with cracks again embedded in the lower halfspace having the same stiffnesses as in model one.
Both examples differ by the saturation of the upper halfspace only.
It occurs that the magnitude of $R_{pp}$ at $e=0$ is positive so that the expected pattern is different as compared to the previous example.
However, crack densities $e_{irr}$ and $e_{rr}$ stay the same.
The discussed results show that halfspace's saturation influences the magnitude of azimuthal variations, but not the shape.
This phenomenon can be explained easily.
The magnitude of the reflection coefficient changes since the isotropic term $R_{ipp}$ depends on both halfspaces.
On the other hand, $\Delta R_{pp}$ from expression~(\ref{DRpp})--- responsible for CAVA shape---depends on the stiffnesses from one halfspace only (identical for both models).
In Table~\ref{tab:15}, we present thirty distinct halfspaces; hence, we provide thirty independent measures of $e_{irr}$ and $e_{rr}$\,. 
\begin{table}[!htbp]
\scalebox{0.99}{
\begin{tabular}
{cccccccc}
\toprule
model & \multirow{1}{*}{halfspaces} & $R_{pp}$&$R_{pp}$ mono-&expected & actual &\multirow{2}{*}{$e_{ir}$} & \multirow{2}{*}{$e_{rr}$}\\
  nr&(upper/lower)& at $e=0$ &tonicity&pattern & pattern &   & \\
 \cmidrule{1-8}
\multirow{2}{*}{1}&b. sand (E5)/&\multirow{2}{*}{negative}&increasing&Fig.~\ref{fig:pattern2}&Fig.~\ref{fig:pattern2a}$^*$&--- & --- \\
&\,\,\,b. sand (E2) &&decreasing& Fig.~\ref{fig:pattern1} &Fig.~\ref{fig:pattern1b}$^*$&$0.16$ &0.20 \\
\cmidrule{1-8}
\multirow{2}{*}{2}&g. sand (E5)/&\multirow{2}{*}{positive}&non mono.&Fig.~\ref{fig:pattern1}&Fig.~\ref{fig:pattern1a}$^*$&--- & --- \\
&\,\,\,b. sand (E2) &&decreasing& Fig.~\ref{fig:pattern2} &Fig.~\ref{fig:pattern2b}$^*$&0.16 &0.20 \\
\cmidrule{1-8}
\multirow{2}{*}{3}&b. limestone (1)/&\multirow{2}{*}{negative}&increasing&Fig.~\ref{fig:pattern2}&Fig.~\ref{fig:pattern2b}&$0.17$ & 0.22 \\
&\,\,\,b. limestone (2) &&decreasing&Fig.~\ref{fig:pattern1}&Fig.~\ref{fig:pattern1b}$^*$&$0.12$ &0.15 \\
\cmidrule{1-8}
\multirow{2}{*}{4}&g. limestone (1)/&\multirow{2}{*}{positive}&increasing&Fig.~\ref{fig:pattern1}&Fig.~\ref{fig:pattern1b}$^*$&$0.10$ & $0.13$ \\
&\,\,\,b. limestone (2) &&decreasing&Fig.~\ref{fig:pattern2} &Fig.~\ref{fig:pattern2b}$^*$&$0.12$ &0.15\\
\cmidrule{1-8}
\multirow{2}{*}{5}&b. shale (B1)/&\multirow{2}{*}{positive}&increasing&Fig.~\ref{fig:pattern1} &Fig.~\ref{fig:pattern1b}$^*$&0.14& 0.18\\
&\,\,\,b. shale (B2)&&non mono.&Fig.~\ref{fig:pattern2} &Fig.~\ref{fig:pattern2a}$^*$&---&---  \\
\cmidrule{1-8}
\multirow{2}{*}{6}&b. shale (G3)/&\multirow{2}{*}{positive}&non mono.&Fig.~\ref{fig:pattern1} &Fig.~\ref{fig:pattern1a}$^*$& ---& --- \\
&\,\,\,b. shale (G5)&&non mono.&Fig.~\ref{fig:pattern2} &Fig.~\ref{fig:pattern2a}$^*$&---&--- \\
\cmidrule{1-8}
\multirow{2}{*}{7}&b. shale (E1)/&\multirow{2}{*}{positive}&increasing&Fig.~\ref{fig:pattern1}&Fig.~\ref{fig:pattern1b}$^*$& 0.78&$>1$\\
&\,\,\,b. shale (E5)&&decreasing&Fig.~\ref{fig:pattern2}&Fig.~\ref{fig:pattern2b}$^*$& $>1$&--- \\
\cmidrule{1-8}
\multirow{2}{*}{8}&b. sand (E5)/&\multirow{2}{*}{positive}&non mono.&Fig.~\ref{fig:pattern1}&Fig.~\ref{fig:pattern1a}$^*$&---&--- \\
&\,\,\,b. shale (E5)&&decreasing&Fig.~\ref{fig:pattern2}&Fig.~\ref{fig:pattern2b}$^*$&$>1$ & ---\\
\cmidrule{1-8}
\multirow{2}{*}{9}&g. sand (E5)/&\multirow{2}{*}{positive}&non mono.&Fig.~\ref{fig:pattern1}&Fig.~\ref{fig:pattern1a}$^*$&--- &---  \\
&\,\,\,b. shale (E5)&&decreasing&Fig.~\ref{fig:pattern2}&Fig.~\ref{fig:pattern2b}$^*$&$>1$&---  \\
\cmidrule{1-8}
\multirow{2}{*}{10}&g. sand (G8)/&\multirow{2}{*}{positive}&non mono.&Fig.~\ref{fig:pattern1}&Fig.~\ref{fig:pattern1a}$^*$& ---& --- \\
&\,\,\,b. sand (G8)&&decreasing&Fig.~\ref{fig:pattern2}&Fig.~\ref{fig:pattern2a}$^*$&---&---\\
\cmidrule{1-8}
\multirow{2}{*}{11}&g. sand (G14)/&\multirow{2}{*}{negative}&non mono.&Fig.~\ref{fig:pattern2}&Fig.~\ref{fig:pattern2a}$^*$&---&---\\
&\,\,\,g. sand (G16)&&non mono.&Fig.~\ref{fig:pattern1}&Fig.~\ref{fig:pattern1a}$^*$& ---& --- \\
\cmidrule{1-8}
\multirow{2}{*}{12}&g. coal (G31)/&\multirow{2}{*}{positive}&non mono.&Fig.~\ref{fig:pattern1}&Fig.~\ref{fig:pattern1b}$^*$& 0.04&0.05  \\
&\,\,\,b. coal (G31)&&non mono.&Fig.~\ref{fig:pattern2}&Fig.~\ref{fig:pattern2b}$^*$& 0.31& 0.48 \\
\cmidrule{1-8}
\multirow{2}{*}{13}&g. limestone (9)/&\multirow{2}{*}{positive}&increasing&Fig.~\ref{fig:pattern1}&Fig.~\ref{fig:pattern1a}$^*$&---&---  \\
&\,\,\,g. limestone (10)&&decreasing&Fig.~\ref{fig:pattern2}&Fig.~\ref{fig:pattern2b}$^*$&0.10 &0.13  \\
\cmidrule{1-8}
\multirow{2}{*}{14}&b. limestone (9)/&\multirow{2}{*}{positive}&increasing&Fig.~\ref{fig:pattern1}&Fig.~\ref{fig:pattern1b}$^*$&0.10 &0.12 \\
&\,\,\,g. limestone (10)&&decreasing&Fig.~\ref{fig:pattern2}&Fig.~\ref{fig:pattern2b}&0.10 &0.13  \\
\cmidrule{1-8}
\multirow{2}{*}{15}&b. limestone (22)/&\multirow{2}{*}{positive}&increasing&Fig.~\ref{fig:pattern1}&Fig.~\ref{fig:pattern1b}$^*$& 0.20& 0.26 \\
&\,\,\,b. dolomite (23)&&decreasing&Fig.~\ref{fig:pattern2}&Fig.~\ref{fig:pattern2b}$^*$&0.10  &0.13  \\
\cmidrule{1-8}
\multirow{2}{*}{16}&g. limestone (22)/&\multirow{2}{*}{positive}&increasing&Fig.~\ref{fig:pattern1}&Fig.~\ref{fig:pattern1b}$^*$& 0.06&0.08   \\
&\,\,\,b. dolomite (23)&&decreasing&Fig.~\ref{fig:pattern2}&Fig.~\ref{fig:pattern2b}$^*$& 0.10&0.13  \\
\cmidrule{1-8}
\multirow{2}{*}{17}&g. dolomite (28)/&\multirow{2}{*}{negative}&increasing&Fig.~\ref{fig:pattern2}&Fig.~\ref{fig:pattern2b}& 0.14 & 0.18\\
&\,\,\,g. dolomite (29)&&decreasing&Fig.~\ref{fig:pattern1}&Fig.~\ref{fig:pattern1a}$^*$&---& ---  \\
\cmidrule{1-8}
\multirow{2}{*}{18}&b. dolomite (28)/&\multirow{2}{*}{negative}&increasing&Fig.~\ref{fig:pattern2}&Fig.~\ref{fig:pattern2b}$^*$& 0.12& 0.15\\
&\,\,\,g. dolomite (29)&&decreasing& Fig.~\ref{fig:pattern1}&Fig.~\ref{fig:pattern1a}$^*$&--- &---  \\
\cmidrule{1-8}
\multirow{2}{*}{19}&b. dolomite (31)/&\multirow{2}{*}{negative}&increasing&Fig.~\ref{fig:pattern2}&Fig.~\ref{fig:pattern2b}&0.15 &0.19 \\
&\,\,\,b. limestone (32)&&decreasing&Fig.~\ref{fig:pattern1}&Fig.~\ref{fig:pattern1b}$^*$& 0.35&0.49  \\
\cmidrule{1-8}
\multirow{2}{*}{20}&g. dolomite (31)/&\multirow{2}{*}{positive}&increasing&Fig.~\ref{fig:pattern1}&Fig.~\ref{fig:pattern1b}$^*$&0.09&0.11  \\
&\,\,\,b. limestone (32)&&decreasing&Fig.~\ref{fig:pattern2}&Fig.~\ref{fig:pattern2b}$^*$&0.35 & 0.49 \\
\bottomrule
\end{tabular}
}
\caption[\small{Verification of the conjectured CAVA patterns: $15^\circ$ of incidence }]{\footnotesize{To verify the conjectured CAVA patterns, we propose twenty models of cracked media. Each model has embedded cracks in either upper or lower background, so that approximated critical density parameters ($e_{ir}$ and $e_{rr}$) for each possibility are obtained (forty cases in total). Backgrounds are brine (b.) or gas (g.) saturated. An asterisk indicates a shortened pattern not shown explicitly in Figures~\ref{fig:pattern1} and~\ref{fig:pattern2}. 
A small/moderate incidence angle, $\theta=15^\circ$\,, is chosen. }}
\label{tab:15}
\end{table}

Based on two hundred eighty examples from Table~\ref{tab:15} and Appendix~\ref{ap:ch8_one}, we infer that our conjectured, full, and shortened CAVA patterns---in a great majority of cases---are correct.
In one example only (for $\theta=30^\circ$), an alternative pattern is needed.
It was caused by predominantly increasing (instead of decreasing) $R_{pp}^{\ell}$ with growing $e$\,.
Our conjectured patterns are also sufficient in other examples of non-monotonic $\partial_eR_{pp}$ (see the fourth column).
To answer the rest of the questions posed at the beginning of this section, we propose to condense the numerical results in Figures~\ref{fig:short}--\ref{fig:e}.
  
Figure~\ref{fig:short} shows that in the majority of cases, CAVA patterns are shortened.
Full patterns are more probable for larger incidences than small ones.
Also, it illustrates that patterns from Figures~\ref{fig:pattern1a} and~\ref{fig:pattern2a} are less likely to occur than patterns from Figures~\ref{fig:pattern1b} and~\ref{fig:pattern2b}.
As we have expected in the previous section, Figures~\ref{fig:pattern1a} and~\ref{fig:pattern2a} are typical for small $\theta$\,.
They do not occur for very large incidences.
\begin{figure}[!htbp]
\centering
\includegraphics[width=0.55\textwidth]{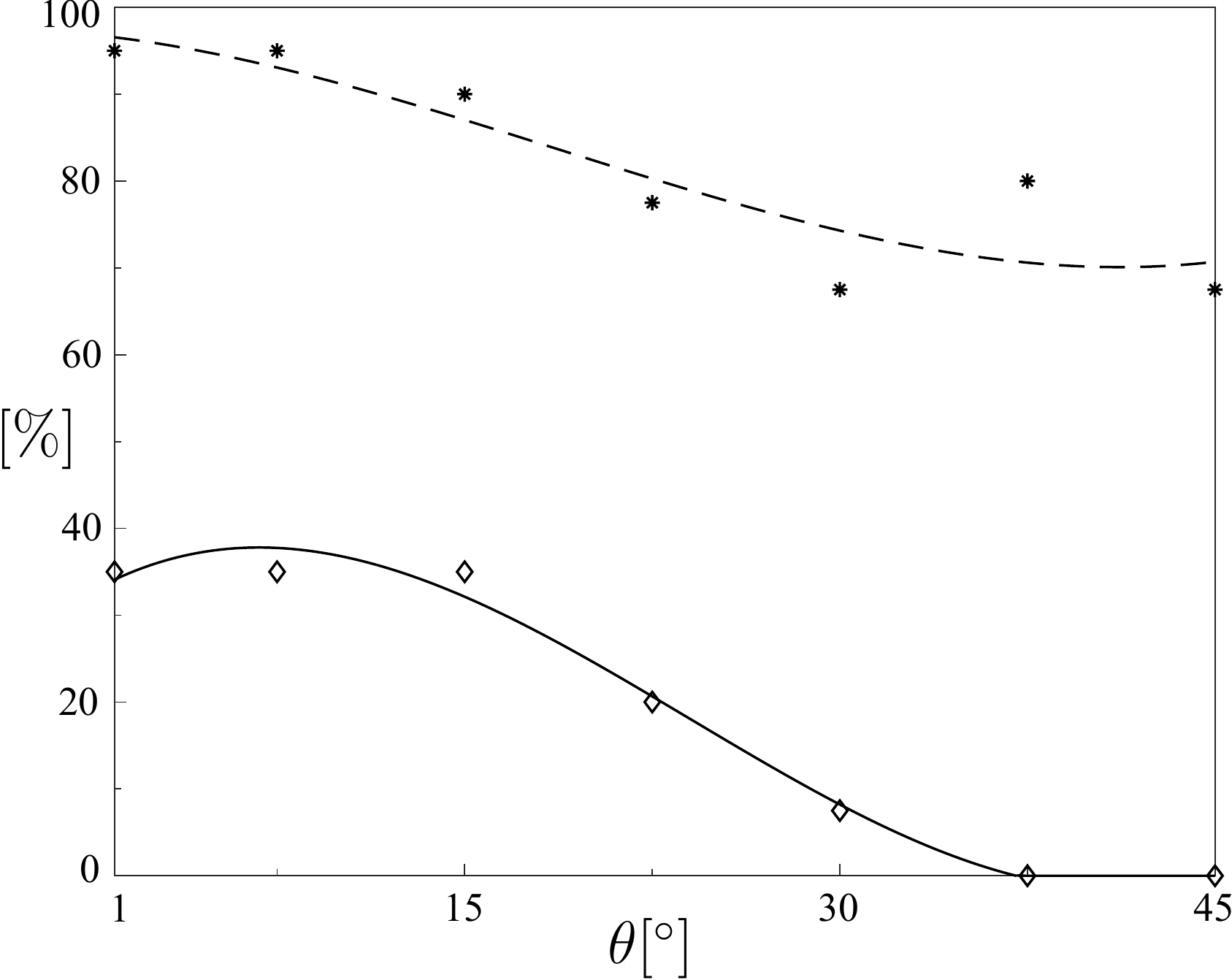}
\caption[\small{Percentage of cracked halfspaces: Shortened patterns followed}]{\small{Asterisks indicate the percentage of cracked halfspaces, where CAVA follows the shortened patterns. Results are obtained for seven incidence angles, and a dashed line shows the trend for $\theta\in[1^\circ,45^\circ]$\,.
Diamonds indicate the percentage of cracked halfspaces, where CAVA  follows the pattern from Figures~\ref{fig:pattern1a} or~\ref{fig:pattern2a}. 
A solid line proposes a possible trend.
}}
\label{fig:short}
\end{figure}
\begin{figure}[!htbp]
\centering
\includegraphics[width=0.55\textwidth]{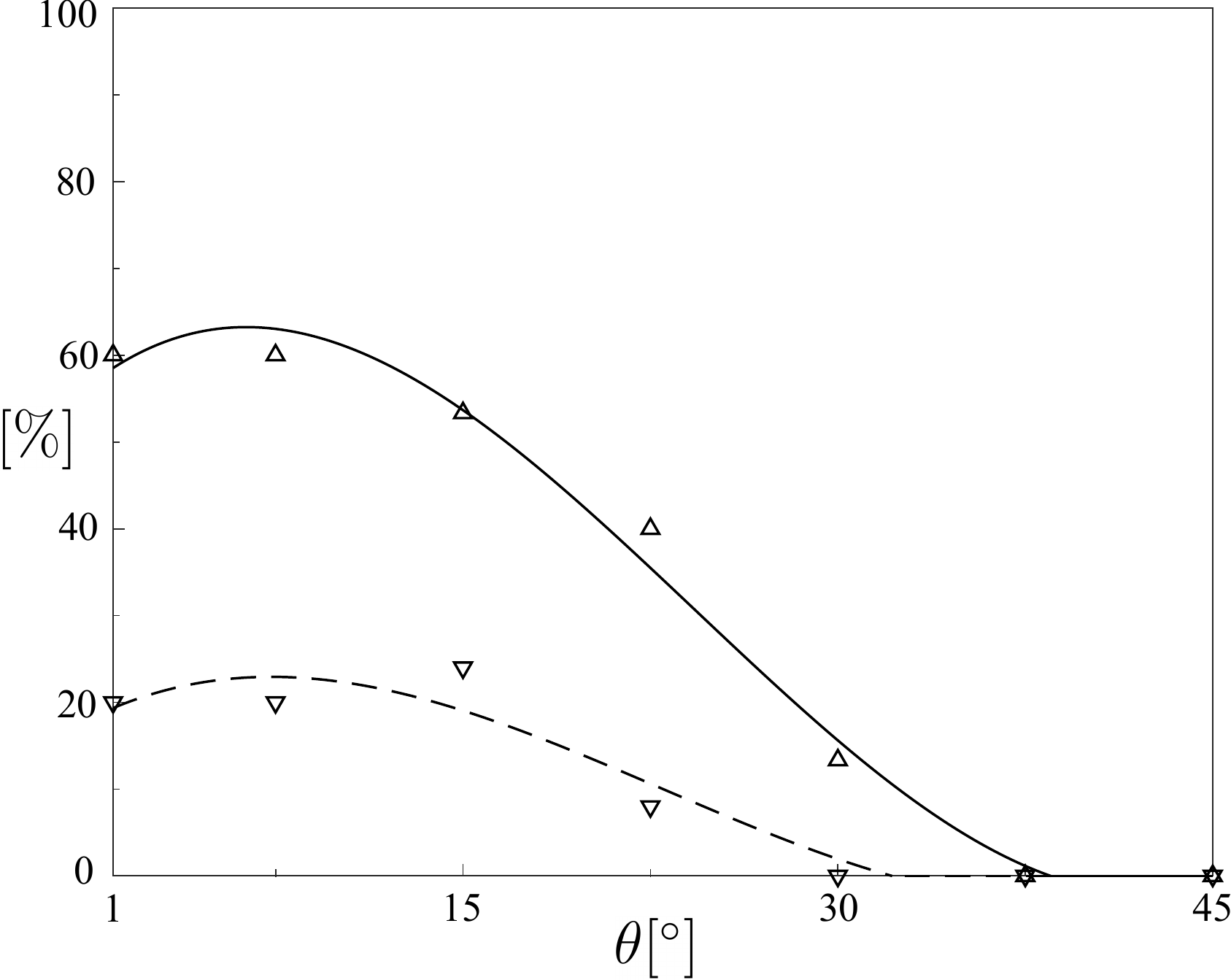}
\caption[\small{Percentage of gas or brine-saturated halfspaces: Figures~\ref{fig:pattern1a} or~\ref{fig:pattern2a} followed}]{\small{Upward pointing triangles indicate the percentage of gas-saturated cracked halfspaces, where CAVA follows the pattern from Figures~\ref{fig:pattern1a} or~\ref{fig:pattern2a}. Results are obtained for seven incidence angles, and a solid line shows the trend for $\theta\in[1^\circ,45^\circ]$\,.
Analogously, downward-pointing triangles correspond to brine-saturated cracked halfspaces, whereas a dashed line is the trend.}}
\label{fig:gas}
\end{figure}
\begin{figure}[!htbp]
\centering
\begin{subfigure}{.48\textwidth}
  \centering
\includegraphics[width=1\textwidth]{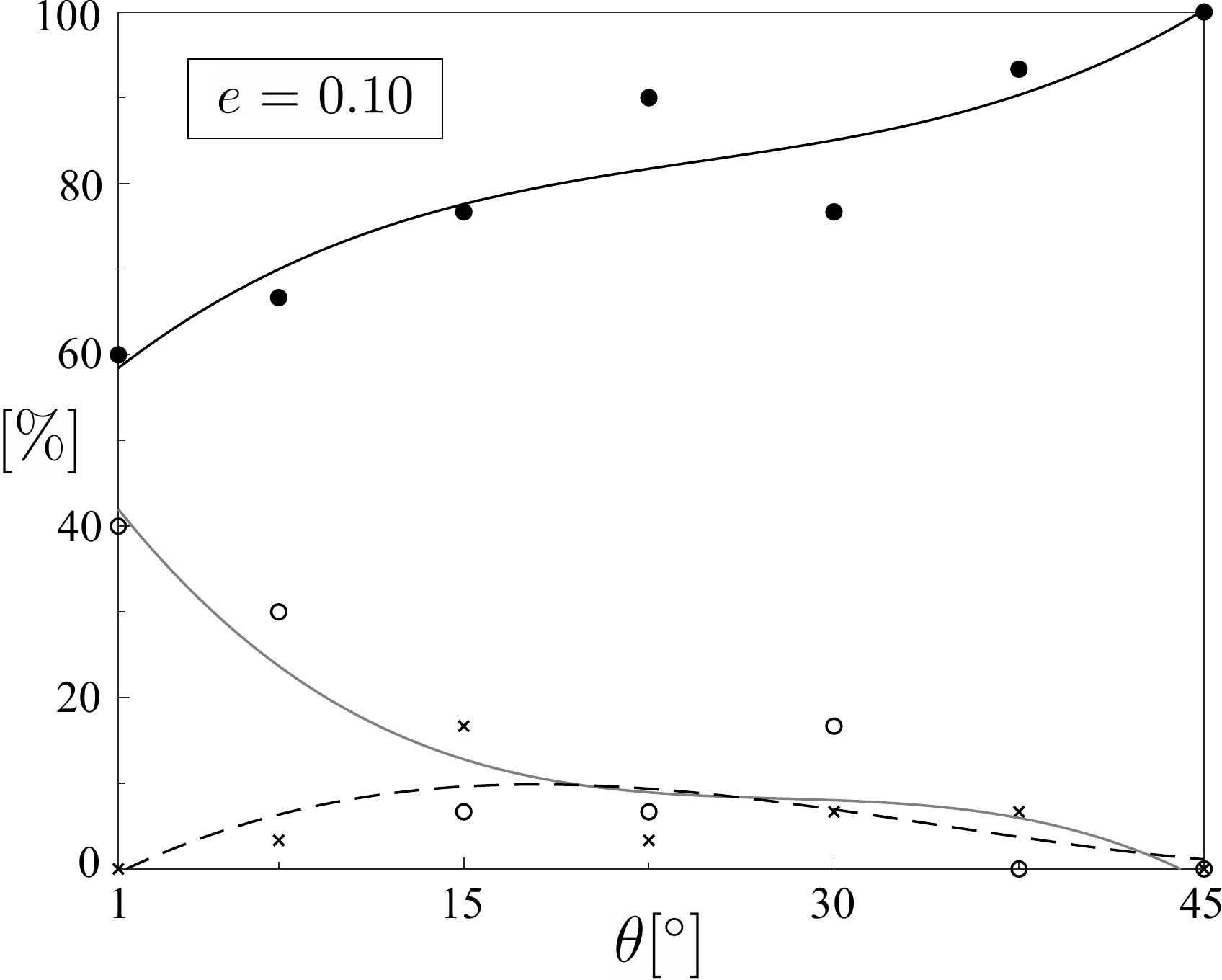}
\end{subfigure}
\quad
\begin{subfigure}{.48\textwidth}
  \centering
\includegraphics[width=1\textwidth]{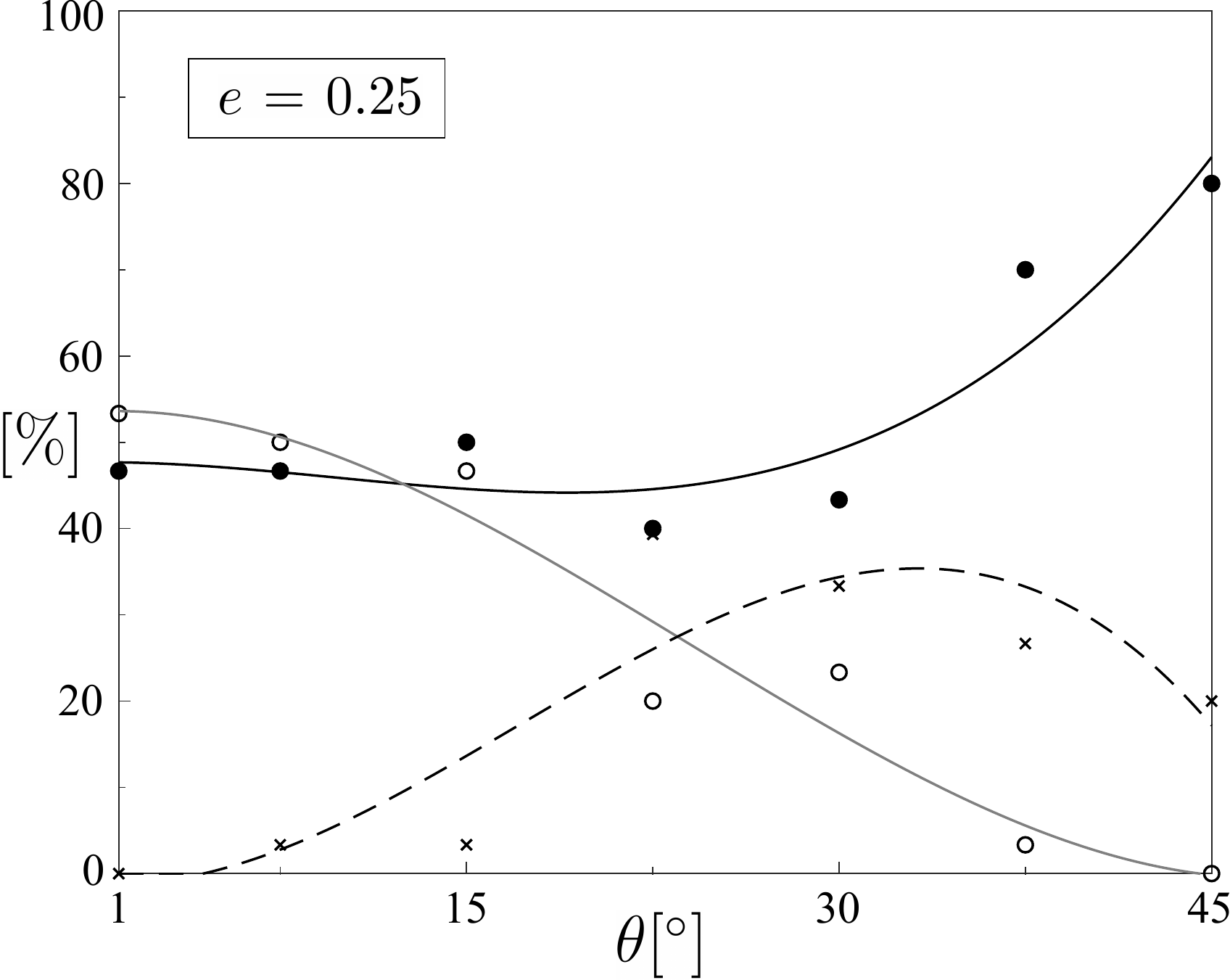}
\end{subfigure}
\caption[\small{Percentage of cracked halfspaces: Regular, irregular, reversed azimuthal variations}]{\small{Percentage of cracked halfspaces, where CAVA is regular (filled circles), irregular (crosses), or reversed (empty circles). Either moderate concentration of cracks $e=0.10$ (graph on the left) or large crack density $e=0.25$ (graph on the right) is assumed.
Results are obtained for seven specific incidences. Possible trends are proposed; solid black line corresponds to the regular phase, the dashed line indicates the irregularity, and the solid grey line shows the reversed phase. 
}}
\label{fig:e}
\end{figure}

Figure~\ref{fig:gas} illustrates that for small angles, sixty percent of gas-bearing and cracked halfspaces exhibit CAVA patterns from Figures~\ref{fig:pattern1a} and~\ref{fig:pattern2a}.  
This percentage is much smaller for brine-saturated rocks.
For both saturations, the percentage decreases with increasing incidence.
Knowing that cracks are embedded in a brine-saturated background, we can expect patterns from Figures~\ref{fig:pattern1b} and~\ref{fig:pattern2b} (for any $\theta$).
On the other hand, if we know that cracks are embedded in a gas-saturated background, for large $\theta$ we should expect Figures~\ref{fig:pattern1b} and~\ref{fig:pattern2b}, but for small incidence, any pattern is probable.
Considering the above, if CAVA belongs to patterns from Figures~\ref{fig:pattern1a} and~\ref{fig:pattern2a}, we should expect that rocks are saturated by gas.

To answer the last question posed in this section, we present Figure~\ref{fig:e}.
It shows that with growing $e$, irregular and reversed phases become more frequent, but the regular phase is usually predominant (except small incidence and large crack concentration).
In general, the irregular phase is the most frequent for moderate angles $\theta\in[15^\circ,30^\circ]$\,, whereas the reversed phase for very small incidences.
For $e=0.10$ irregular phase may be present up to every sixth case $(\theta=15^\circ)$.
Reversed CAVA can be demonstrated by up to forty percent interfaces $(\theta=1^\circ)$\,.
For $e=0.25$ irregularity may occur in up to fourty percent of cases ($\theta=22.5^\circ$)\,, whereas the reversed phase in more than fifty percent of examples $(\theta=1^\circ)$\,.
Larger concentrations than $e=0.25$ are not illustrated due to the dubious accuracy of the NIA in such cases.

Having verified our patterns and answered all the essential questions regarding the nature of azimuthal variations, let us summarise the key findings regarding gas exploration that interest many geophysicists dealing with cracked media.
First, the saturation of cracked media changes the magnitude of variations, but not its shape.
Second, knowing about the presence of gas, we cannot infer the right CAVA shape.
Third, knowing CAVA that magnitude and shape is specific only for Figures~\ref{fig:pattern1a} and~\ref{fig:pattern2a}, with significant probability, we can expect gas saturation.
Fourth, patterns from the aforementioned figures do not occur for large incidence angles.

At the end of the previous section, we have mentioned the example of CAVA attributes (sign and shape) characteristic for gas-bearing rocks only.
Herein, we extract all CAVA that appear in Figures~\ref{fig:pattern1a} and~\ref{fig:pattern2a}, but are absent in~Figures~\ref{fig:pattern1b} and~\ref{fig:pattern2b}.
Therefore,
Figure~\ref{fig:signature} gathers all variations characteristic for gas presence.
Again, the existence of gas-bearing rocks does not assure these shapes.
However, CAVA from Figure~\ref{fig:signature} can be treated as a gas indicator.
\begin{figure}[!htbp]
\includegraphics[width=1\textwidth]{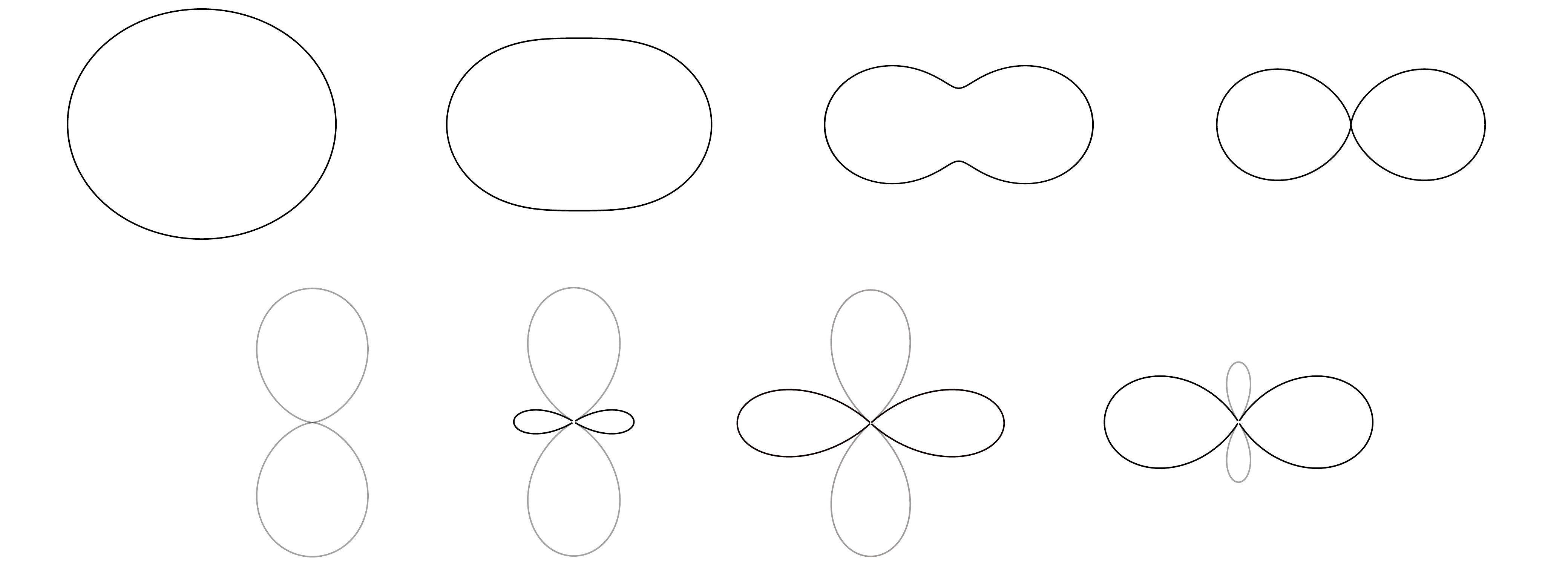}
\caption[\small{Azimuthal variations of amplitude characteristic for gas-bearing rocks}]{\small{Azimuthal variations of amplitude characteristic for gas-bearing rocks. The set of cracks is vertical and normal to the $x_1$-axis. If cracks are embedded in the lower halfspace, then black and grey colours correspond to positive and negative reflection coefficient, respectively. If cracks are situated in the upper halfspace, the meaning of colours is the opposite. }}
\label{fig:signature}
\end{figure} 
%
%
\section{Conclusions}
We have analysed the effect of crack concentration on the PP-wave reflection coefficient variations with azimuth.
Such effect differs depending on the incidence angle and stiffnesses of the cracked medium (influenced by the rock saturation).
We have assumed a single set of vertically aligned cracks with a normal to the $x_1$-axis.
We have examined cracks embedded in one halfspace only, either isotropic or anisotropic (VTI), employing the effective medium theory.

We have proposed and verified patterns of two-dimensional azimuthal variations of amplitude changing with increasing crack concentration upon thorough analytical and numerical analysis.
We have recognised patterns typical for small incidence and gas saturation, and schemes characteristic for large incidence and brine-bearing rocks.
Certain azimuthal variations (sign and shape) are present solely in the patterns typical for gas saturation.
We have indicated eight shapes characteristic for cracks situated in the gas-bearing halfspace. 

We have also noticed that the reflection coefficient may have extreme absolute values in directions other than parallel or perpendicular to cracks. An irregular variation occurs in such cases, which is more frequent for moderate incidences and large crack concentration. 

We are aware of the limitations imposed on our findings.
Vavry\v{c}uk-P\v{s}en\v{c}\'{i}k  approximation of the PP-wave reflection coefficient that we use assumes weak anisotropy and weak elastic contrasts at the interfaces.
Such simplifications are needed to perform a fruitful analytical analysis.
Moreover, we assume the non-interactive approximation that is inaccurate for larger concentrations of cracks.
Patterns proposed by us are valid for cracks with a normal parallel to the $x_1$-axis only.
However, analogical patterns can be obtained for other orientations of cracks, using methods from this paper.
The shapes from our patterns should be rotated by the angle equal to the deviation of cracks from the $x_1$-axis.
Further, we expect that the CAVA effect caused by several sets of cracks is a kind of superposition of patterns corresponding to each set.
In the future, we aim to verify our anticipations.
Also, we intend to provide real data examples to examine the findings and conjectures shown herein.  %
\section*{Acknowledgements}
We wish to acknowledge discussions with Michael A. Slawinski.
The research was done in the context of The Geomechanics Project partially supported by the Natural Sciences and Engineering Research Council of Canada, grant 202259.

\bibliography{bibliography}
\bibliographystyle{apa}
\begin{appendix}
\newpage
\section{Tables with numerical results}\label{ap:ch8_one}
\begin{table}[!htbp]
\scalebox{0.9}{
\begin{tabular}
{ccccccc}
 \cmidrule{2-7}
&\multicolumn{6}{c}{$\theta=1^\circ$}\\
\toprule
 \multirow{2}{*}{models} & $R_{pp}$&$R_{pp}$ mono-&expected & actual &\multirow{2}{*}{$e_{ir}$} & \multirow{2}{*}{$e_{rr}$}\\
  & at $e=0$ &tonicity&pattern & pattern &   & \\
 \cmidrule{1-7}
b. sand (E5)/&\multirow{2}{*}{negative}&increasing&Fig.~\ref{fig:pattern2a}&Fig.~\ref{fig:pattern2a}$^*$&--- & --- \\
\,\,\,b. sand (E2) &&decreasing& Fig.~\ref{fig:pattern1a} &Fig.~\ref{fig:pattern1b}$^*$&$0.12$ &0.12 \\
\cmidrule{1-7}
g. sand (E5)/&\multirow{2}{*}{positive}&increasing&Fig.~\ref{fig:pattern1a}&Fig.~\ref{fig:pattern1a}$^*$&--- & --- \\
\,\,\,b. sand (E2) &&decreasing& Fig.~\ref{fig:pattern2a} &Fig.~\ref{fig:pattern2b}$^*$&0.12 &0.12 \\
\cmidrule{1-7}
b. limestone (1)/&\multirow{2}{*}{negative}&increasing&Fig.~\ref{fig:pattern2a}&Fig.~\ref{fig:pattern2b}$^*$&$0.01$ & 0.01 \\
\,\,\,b. limestone (2) &&decreasing&Fig.~\ref{fig:pattern1a}&Fig.~\ref{fig:pattern1b}$^*$&$0.08$ &0.08 \\
\cmidrule{1-7}
g. limestone (1)/&\multirow{2}{*}{positive}&increasing&Fig.~\ref{fig:pattern1a}&Fig.~\ref{fig:pattern1b}$^*$&$0.06$ & $0.06$ \\
\,\,\,b. limestone (2) &&decreasing&Fig.~\ref{fig:pattern2a} &Fig.~\ref{fig:pattern2b}$^*$&$0.08$ &0.08\\
\cmidrule{1-7}
b. shale (B1)/&\multirow{2}{*}{positive}&increasing&Fig.~\ref{fig:pattern1a} &Fig.~\ref{fig:pattern1b}$^*$&0.09& 0.09\\
\,\,\,b. shale (B2)&& decreasing&Fig.~\ref{fig:pattern2a} &Fig.~\ref{fig:pattern2a}$^*$&---&---  \\
\cmidrule{1-7}
b. shale (G3)/&\multirow{2}{*}{positive}&increasing&Fig.~\ref{fig:pattern1a} &Fig.~\ref{fig:pattern1a}$^*$& ---& --- \\
\,\,\,b. shale (G5)&&decreasing&Fig.~\ref{fig:pattern2a} &Fig.~\ref{fig:pattern2a}$^*$&---&--- \\
\cmidrule{1-7}
b. shale (E1)/&\multirow{2}{*}{positive}&increasing&Fig.~\ref{fig:pattern1a}&Fig.~\ref{fig:pattern1b}$^*$& 0.68&0.68\\
\,\,\,b. shale (E5)&&decreasing&Fig.~\ref{fig:pattern2a}&Fig.~\ref{fig:pattern2b}& $>1$&$>1$ \\
\cmidrule{1-7}
b. sand (E5)/&\multirow{2}{*}{negative}&increasing&Fig.~\ref{fig:pattern2a}&Fig.~\ref{fig:pattern2a}$^*$&---&--- \\
\,\,\,b. shale (E5)&&decreasing&Fig.~\ref{fig:pattern1a}&Fig.~\ref{fig:pattern1b}$^*$&$>1$ & $>1$\\
\cmidrule{1-7}
g. sand (E5)/&\multirow{2}{*}{positive}&increasing&Fig.~\ref{fig:pattern1a}&Fig.~\ref{fig:pattern1a}$^*$&--- &---  \\
\,\,\,b. shale (E5)&&decreasing&Fig.~\ref{fig:pattern2a}&Fig.~\ref{fig:pattern2b}$^*$&$>1$&$>1$  \\
\cmidrule{1-7}
g. sand (G8)/&\multirow{2}{*}{positive}&increasing&Fig.~\ref{fig:pattern1a}&Fig.~\ref{fig:pattern1a}$^*$& ---& --- \\
\,\,\,b. sand (G8)&&decreasing&Fig.~\ref{fig:pattern2a}&Fig.~\ref{fig:pattern2b}$^*$&0.04&0.04\\
\cmidrule{1-7}
g. sand (G14)/&\multirow{2}{*}{negative}&increasing&Fig.~\ref{fig:pattern2a}&Fig.~\ref{fig:pattern2a}$^*$&---&---\\
\,\,\,g. sand (G16)&&decreasing&Fig.~\ref{fig:pattern1a}&Fig.~\ref{fig:pattern1a}$^*$& ---& --- \\
\cmidrule{1-7}
g. coal (G31)/&\multirow{2}{*}{positive}&increasing&Fig.~\ref{fig:pattern1a}&Fig.~\ref{fig:pattern1a}$^*$& ---&---  \\
\,\,\,b. coal (G31)&&decreasing&Fig.~\ref{fig:pattern2a}&Fig.~\ref{fig:pattern2b}$^*$& 0.23& 0.23 \\
\cmidrule{1-7}
g. limestone (9)/&\multirow{2}{*}{positive}&increasing&Fig.~\ref{fig:pattern1a}&Fig.~\ref{fig:pattern1a}$^*$&---&---  \\
\,\,\,g. limestone (10)&&decreasing&Fig.~\ref{fig:pattern2a}&Fig.~\ref{fig:pattern2b}$^*$&0.06 &0.06  \\
\cmidrule{1-7}
b. limestone (9)/&\multirow{2}{*}{positive}&increasing&Fig.~\ref{fig:pattern1a}&Fig.~\ref{fig:pattern1b}$^*$&0.05 &0.05 \\
\,\,\,g. limestone (10)&&decreasing&Fig.~\ref{fig:pattern2a}&Fig.~\ref{fig:pattern2b}$^*$&0.06 &0.06  \\
\cmidrule{1-7}
b. limestone (22)/&\multirow{2}{*}{positive}&increasing&Fig.~\ref{fig:pattern1a}&Fig.~\ref{fig:pattern1b}$^*$& 0.15& 0.15 \\
\,\,\,b. dolomite (23)&&decreasing&Fig.~\ref{fig:pattern2a}&Fig.~\ref{fig:pattern2b}$^*$&0.06  &0.06  \\
\cmidrule{1-7}
g. limestone (22)/&\multirow{2}{*}{positive}&increasing&Fig.~\ref{fig:pattern1a}&Fig.~\ref{fig:pattern1b}$^*$& 0.02&0.02   \\
\,\,\,b. dolomite (23)&&decreasing&Fig.~\ref{fig:pattern2a}&Fig.~\ref{fig:pattern2b}$^*$& 0.06&0.06  \\
\cmidrule{1-7}
g. dolomite (28)/&\multirow{2}{*}{negative}&increasing&Fig.~\ref{fig:pattern2a}&Fig.~\ref{fig:pattern2b}& 0.10 & 0.10\\
\,\,\,g. dolomite (29)&&decreasing&Fig.~\ref{fig:pattern1a}&Fig.~\ref{fig:pattern1a}$^*$&---& ---  \\
\cmidrule{1-7}
b. dolomite (28)/&\multirow{2}{*}{negative}&increasing&Fig.~\ref{fig:pattern2a}&Fig.~\ref{fig:pattern2b}$^*$& 0.08& 0.08\\
\,\,\,g. dolomite (29)&&decreasing& Fig.~\ref{fig:pattern1a}&Fig.~\ref{fig:pattern1a}$^*$&--- &---  \\
\cmidrule{1-7}
b. dolomite (31)/&\multirow{2}{*}{negative}&increasing&Fig.~\ref{fig:pattern2a}&Fig.~\ref{fig:pattern2b}$^*$&0.11 &0.11 \\
\,\,\,b. limestone (32)&&decreasing&Fig.~\ref{fig:pattern1a}&Fig.~\ref{fig:pattern1b}$^*$& 0.29&0.29  \\
\cmidrule{1-7}
g. dolomite (31)/&\multirow{2}{*}{positive}&increasing&Fig.~\ref{fig:pattern1a}&Fig.~\ref{fig:pattern1b}$^*$&0.05&0.05  \\
\,\,\,b. limestone (32)&&decreasing&Fig.~\ref{fig:pattern2a}&Fig.~\ref{fig:pattern2b}$^*$&0.29 & 0.29 \\
\bottomrule
\end{tabular}}
\caption[\small{Verification of the conjectured CAVA patterns: Various incidence angles}]{\footnotesize{To verify the conjectured CAVA patterns, we propose twenty models of cracked media. Each model has embedded cracks in either upper or lower background, so that approximated critical density parameters ($e_{ir}$ and $e_{rr}$) for each possibility are obtained. Backgrounds are brine (b.) or gas (g.) saturated. An asterisk indicates a shortened pattern not shown explicitly in Figures~\ref{fig:pattern1} and~\ref{fig:pattern2}. 
Various incidences are chosen, namely, $\theta=1^\circ$\,, $\theta=7.5^\circ\,$, $\theta=22.5^\circ\,$, $\theta=30^\circ\,$, $\theta=37.5^\circ$\,, and $\theta=45^\circ$\,.  }}
\label{tab:1}
\end{table}

\begin{table}[!htbp]
\scalebox{0.99}{
\begin{tabular}
{ccccccc}
\cmidrule{2-7}
&\multicolumn{6}{c}{$\theta=7.5^\circ$}\\
\toprule
 \multirow{2}{*}{models} & $R_{pp}$&$R_{pp}$ mono-&expected & actual &\multirow{2}{*}{$e_{ir}$} & \multirow{2}{*}{$e_{rr}$}\\
  & at $e=0$ &tonicity&pattern & pattern &   & \\
 \cmidrule{1-7}
b. sand (E5)/&\multirow{2}{*}{negative}&increasing&Fig.~\ref{fig:pattern2a}&Fig.~\ref{fig:pattern2a}$^*$&--- & --- \\
\,\,\,b. sand (E2) &&decreasing& Fig.~\ref{fig:pattern1a} &Fig.~\ref{fig:pattern1b}$^*$&$0.12$ &0.13 \\
\cmidrule{1-7}
g. sand (E5)/&\multirow{2}{*}{positive}&increasing&Fig.~\ref{fig:pattern1a}&Fig.~\ref{fig:pattern1a}$^*$&--- & --- \\
\,\,\,b. sand (E2) &&decreasing& Fig.~\ref{fig:pattern2a} &Fig.~\ref{fig:pattern2b}$^*$&0.12 &0.13 \\
\cmidrule{1-7}
b. limestone (1)/&\multirow{2}{*}{negative}&increasing&Fig.~\ref{fig:pattern2a}&Fig.~\ref{fig:pattern2b}$^*$&$0.13$ & 0.14 \\
\,\,\,b. limestone (2) &&decreasing&Fig.~\ref{fig:pattern1a}&Fig.~\ref{fig:pattern1b}$^*$&$0.08$ &0.09 \\
\cmidrule{1-7}
g. limestone (1)/&\multirow{2}{*}{positive}&increasing&Fig.~\ref{fig:pattern1a}&Fig.~\ref{fig:pattern1b}$^*$&$0.07$ & $0.08$ \\
\,\,\,b. limestone (2) &&decreasing&Fig.~\ref{fig:pattern2a} &Fig.~\ref{fig:pattern2b}$^*$&$0.08$ &0.09\\
\cmidrule{1-7}
b. shale (B1)/&\multirow{2}{*}{positive}&increasing&Fig.~\ref{fig:pattern1a} &Fig.~\ref{fig:pattern1b}$^*$&0.10& 0.11\\
\,\,\,b. shale (B2)&& decreasing&Fig.~\ref{fig:pattern2a} &Fig.~\ref{fig:pattern2a}$^*$&---&---  \\
\cmidrule{1-7}
b. shale (G3)/&\multirow{2}{*}{positive}&increasing&Fig.~\ref{fig:pattern1a} &Fig.~\ref{fig:pattern1a}$^*$& ---& --- \\
\,\,\,b. shale (G5)&&non mono.&Fig.~\ref{fig:pattern2a} &Fig.~\ref{fig:pattern2a}$^*$&---&--- \\
\cmidrule{1-7}
b. shale (E1)/&\multirow{2}{*}{positive}&increasing&Fig.~\ref{fig:pattern1a}&Fig.~\ref{fig:pattern1b}$^*$& 0.71&0.83\\
\,\,\,b. shale (E5)&&decreasing&Fig.~\ref{fig:pattern2a}&Fig.~\ref{fig:pattern2b}& $>1$&$>1$ \\
\cmidrule{1-7}
b. sand (E5)/&\multirow{2}{*}{negative}&increasing&Fig.~\ref{fig:pattern2a}&Fig.~\ref{fig:pattern2a}$^*$&---&--- \\
\,\,\,b. shale (E5)&&decreasing&Fig.~\ref{fig:pattern1a}&Fig.~\ref{fig:pattern1b}$^*$&$>1$ & $>1$\\
\cmidrule{1-7}
g. sand (E5)/&\multirow{2}{*}{positive}&non mono.&Fig.~\ref{fig:pattern1a}&Fig.~\ref{fig:pattern1a}$^*$&--- &---  \\
\,\,\,b. shale (E5)&&decreasing&Fig.~\ref{fig:pattern2a}&Fig.~\ref{fig:pattern2b}$^*$&$>1$&$>1$  \\
\cmidrule{1-7}
g. sand (G8)/&\multirow{2}{*}{positive}&increasing&Fig.~\ref{fig:pattern1a}&Fig.~\ref{fig:pattern1a}$^*$& ---& --- \\
\,\,\,b. sand (G8)&&decreasing&Fig.~\ref{fig:pattern2a}&Fig.~\ref{fig:pattern2b}$^*$&0.05&0.05\\
\cmidrule{1-7}
g. sand (G14)/&\multirow{2}{*}{negative}&non mono.&Fig.~\ref{fig:pattern2a}&Fig.~\ref{fig:pattern2a}$^*$&---&---\\
\,\,\,g. sand (G16)&&decreasing&Fig.~\ref{fig:pattern1a}&Fig.~\ref{fig:pattern1a}$^*$& ---& --- \\
\cmidrule{1-7}
g. coal (G31)/&\multirow{2}{*}{positive}&non mono.&Fig.~\ref{fig:pattern1a}&Fig.~\ref{fig:pattern1a}$^*$& ---&---  \\
\,\,\,b. coal (G31)&&decreasing&Fig.~\ref{fig:pattern2a}&Fig.~\ref{fig:pattern2b}$^*$& 0.25& 0.28 \\
\cmidrule{1-7}
g. limestone (9)/&\multirow{2}{*}{positive}&increasing&Fig.~\ref{fig:pattern1a}&Fig.~\ref{fig:pattern1a}$^*$&---&---  \\
\,\,\,g. limestone (10)&&decreasing&Fig.~\ref{fig:pattern2a}&Fig.~\ref{fig:pattern2b}$^*$&0.07 &0.07  \\
\cmidrule{1-7}
b. limestone (9)/&\multirow{2}{*}{positive}&increasing&Fig.~\ref{fig:pattern1a}&Fig.~\ref{fig:pattern1b}$^*$&0.06 &0.07 \\
\,\,\,g. limestone (10)&&decreasing&Fig.~\ref{fig:pattern2a}&Fig.~\ref{fig:pattern2b}$^*$&0.07 &0.07  \\
\cmidrule{1-7}
b. limestone (22)/&\multirow{2}{*}{positive}&increasing&Fig.~\ref{fig:pattern1a}&Fig.~\ref{fig:pattern1b}$^*$& 0.16& 0.17 \\
\,\,\,b. dolomite (23)&&decreasing&Fig.~\ref{fig:pattern2a}&Fig.~\ref{fig:pattern2b}$^*$&0.07  &0.07  \\
\cmidrule{1-7}
g. limestone (22)/&\multirow{2}{*}{positive}&increasing&Fig.~\ref{fig:pattern1a}&Fig.~\ref{fig:pattern1b}$^*$& 0.03&0.03   \\
\,\,\,b. dolomite (23)&&decreasing&Fig.~\ref{fig:pattern2a}&Fig.~\ref{fig:pattern2b}$^*$& 0.07&0.07  \\
\cmidrule{1-7}
g. dolomite (28)/&\multirow{2}{*}{negative}&increasing&Fig.~\ref{fig:pattern2a}&Fig.~\ref{fig:pattern2b}& 0.11 & 0.12\\
\,\,\,g. dolomite (29)&&decreasing&Fig.~\ref{fig:pattern1a}&Fig.~\ref{fig:pattern1a}$^*$&---& ---  \\
\cmidrule{1-7}
b. dolomite (28)/&\multirow{2}{*}{negative}&increasing&Fig.~\ref{fig:pattern2a}&Fig.~\ref{fig:pattern2b}$^*$& 0.08& 0.09\\
\,\,\,g. dolomite (29)&&decreasing& Fig.~\ref{fig:pattern1a}&Fig.~\ref{fig:pattern1a}$^*$&--- &---  \\
\cmidrule{1-7}
b. dolomite (31)/&\multirow{2}{*}{negative}&increasing&Fig.~\ref{fig:pattern2a}&Fig.~\ref{fig:pattern2b}$^*$&0.11 &0.12 \\
\,\,\,b. limestone (32)&&decreasing&Fig.~\ref{fig:pattern1a}&Fig.~\ref{fig:pattern1b}$^*$& 0.31&0.33  \\
\cmidrule{1-7}
g. dolomite (31)/&\multirow{2}{*}{positive}&increasing&Fig.~\ref{fig:pattern1a}&Fig.~\ref{fig:pattern1b}$^*$&0.05&0.06  \\
\,\,\,b. limestone (32)&&decreasing&Fig.~\ref{fig:pattern2a}&Fig.~\ref{fig:pattern2b}$^*$&0.31 & 0.33 \\
\bottomrule
\end{tabular}}
\label{tab:75}
\end{table}

\begin{table}[!htbp]
\scalebox{0.99}{
\begin{tabular}
{ccccccc}
\cmidrule{2-7}
&\multicolumn{6}{c}{$\theta=22.5^\circ$}\\
\toprule
 \multirow{2}{*}{models} & $R_{pp}$&$R_{pp}$ mono-&expected & actual &\multirow{2}{*}{$e_{ir}$} & \multirow{2}{*}{$e_{rr}$}\\
  & at $e=0$ &tonicity&pattern & pattern &   & \\
 \cmidrule{1-7}
b. sand (E5)/&\multirow{2}{*}{positive}&non mono.&Fig.~\ref{fig:pattern1}&Fig.~\ref{fig:pattern1b}&0.02 & 0.04\\
\,\,\,b. sand (E2) &&decreasing& Fig.~\ref{fig:pattern2} &Fig.~\ref{fig:pattern2b}&$0.22$ &0.40 \\
\cmidrule{1-7}
g. sand (E5)/&\multirow{2}{*}{positive}&non mono.&Fig.~\ref{fig:pattern1}&Fig.~\ref{fig:pattern1a}$^*$&--- & --- \\
\,\,\,b. sand (E2) &&decreasing& Fig.~\ref{fig:pattern2} &Fig.~\ref{fig:pattern2b}$^*$&0.22 &0.40 \\
\cmidrule{1-7}
b. limestone (1)/&\multirow{2}{*}{negative}&increasing&Fig.~\ref{fig:pattern2}&Fig.~\ref{fig:pattern2b}&$0.23$ & 0.45 \\
\,\,\,b. limestone (2) &&decreasing&Fig.~\ref{fig:pattern1}&Fig.~\ref{fig:pattern1b}$^*$&$0.17$ &0.33 \\
\cmidrule{1-7}
g. limestone (1)/&\multirow{2}{*}{positive}&increasing&Fig.~\ref{fig:pattern1}&Fig.~\ref{fig:pattern1b}$^*$&$0.15$ & $0.28$ \\
\,\,\,b. limestone (2) &&decreasing&Fig.~\ref{fig:pattern2} &Fig.~\ref{fig:pattern2b}&$0.17$ &0.33\\
\cmidrule{1-7}
b. shale (B1)/&\multirow{2}{*}{positive}&increasing&Fig.~\ref{fig:pattern1} &Fig.~\ref{fig:pattern1b}$^*$&0.20& 0.38\\
\,\,\,b. shale (B2)&& non mono.&Fig.~\ref{fig:pattern2} &Fig.~\ref{fig:pattern2a}$^*$&---&---  \\
\cmidrule{1-7}
b. shale (G3)/&\multirow{2}{*}{positive}&non mono.&Fig.~\ref{fig:pattern1} &Fig.~\ref{fig:pattern1b}& 0.07& 0.11 \\
\,\,\,b. shale (G5)&&non mono.&Fig.~\ref{fig:pattern2} &Fig.~\ref{fig:pattern2a}$^*$&---&--- \\
\cmidrule{1-7}
b. shale (E1)/&\multirow{2}{*}{positive}&increasing&Fig.~\ref{fig:pattern1}&Fig.~\ref{fig:pattern1b}$^*$& 0.92&---\\
\,\,\,b. shale (E5)&&decreasing&Fig.~\ref{fig:pattern2}&Fig.~\ref{fig:pattern2b}$^*$& $>1$&--- \\
\cmidrule{1-7}
b. sand (E5)/&\multirow{2}{*}{positive}&non mono.&Fig.~\ref{fig:pattern1}&Fig.~\ref{fig:pattern1b}&0.02&0.04\\
\,\,\,b. shale (E5)&&decreasing&Fig.~\ref{fig:pattern2}&Fig.~\ref{fig:pattern2b}$^*$&$>1$ & ---\\
\cmidrule{1-7}
g. sand (E5)/&\multirow{2}{*}{positive}&non mono.&Fig.~\ref{fig:pattern1a}&Fig.~\ref{fig:pattern1a}$^*$&--- &---  \\
\,\,\,b. shale (E5)&&decreasing&Fig.~\ref{fig:pattern2a}&Fig.~\ref{fig:pattern2b}$^*$&$>1$&---  \\
\cmidrule{1-7}
g. sand (G8)/&\multirow{2}{*}{positive}&non mono.&Fig.~\ref{fig:pattern1}&Fig.~\ref{fig:pattern1a}$^*$& ---& --- \\
\,\,\,b. sand (G8)&&non mono.&Fig.~\ref{fig:pattern2}&Fig.~\ref{fig:pattern2b}$^*$&0.14&0.26\\
\cmidrule{1-7}
g. sand (G14)/&\multirow{2}{*}{negative}&non mono.&Fig.~\ref{fig:pattern2}&Fig.~\ref{fig:pattern2a}$^*$&---&---\\
\,\,\,g. sand (G16)&&non mono.&Fig.~\ref{fig:pattern1}&Fig.~\ref{fig:pattern1a}$^*$& ---& --- \\
\cmidrule{1-7}
g. coal (G31)/&\multirow{2}{*}{positive}&non mono.&Fig.~\ref{fig:pattern1}&Fig.~\ref{fig:pattern1b}& 0.13&0.31  \\
\,\,\,b. coal (G31)&&non mono.&Fig.~\ref{fig:pattern2}&Fig.~\ref{fig:pattern2b}$^*$& 0.41& $>1$ \\
\cmidrule{1-7}
g. limestone (9)/&\multirow{2}{*}{positive}&non mono.&Fig.~\ref{fig:pattern1}&Fig.~\ref{fig:pattern1a}$^*$&---&---  \\
\,\,\,g. limestone (10)&&decreasing&Fig.~\ref{fig:pattern2}&Fig.~\ref{fig:pattern2b}$^*$&0.16 &0.29  \\
\cmidrule{1-7}
b. limestone (9)/&\multirow{2}{*}{positive}&increasing&Fig.~\ref{fig:pattern1}&Fig.~\ref{fig:pattern1b}$^*$&0.15 &0.28 \\
\,\,\,g. limestone (10)&&decreasing&Fig.~\ref{fig:pattern2}&Fig.~\ref{fig:pattern2b}&0.16 &0.29  \\
\cmidrule{1-7}
b. limestone (22)/&\multirow{2}{*}{positive}&increasing&Fig.~\ref{fig:pattern1}&Fig.~\ref{fig:pattern1b}$^*$& 0.26& 0.54 \\
\,\,\,b. dolomite (23)&&decreasing&Fig.~\ref{fig:pattern2}&Fig.~\ref{fig:pattern2b}$^*$&0.16  &0.29  \\
\cmidrule{1-7}
g. limestone (22)/&\multirow{2}{*}{positive}&increasing&Fig.~\ref{fig:pattern1}&Fig.~\ref{fig:pattern1b}$^*$& 0.11&0.19   \\
\,\,\,b. dolomite (23)&&decreasing&Fig.~\ref{fig:pattern2}&Fig.~\ref{fig:pattern2b}$^*$& 0.16&0.29  \\
\cmidrule{1-7}
g. dolomite (28)/&\multirow{2}{*}{negative}&increasing&Fig.~\ref{fig:pattern2}&Fig.~\ref{fig:pattern2b}$^*$& 0.20 & 0.38\\
\,\,\,g. dolomite (29)&&non mono.&Fig.~\ref{fig:pattern1}&Fig.~\ref{fig:pattern1b}$^*$&0.03&0.06\\
\cmidrule{1-7}
b. dolomite (28)/&\multirow{2}{*}{negative}&increasing&Fig.~\ref{fig:pattern2}&Fig.~\ref{fig:pattern2b}$^*$& 0.18& 0.32\\
\,\,\,g. dolomite (29)&&decreasing& Fig.~\ref{fig:pattern1}&Fig.~\ref{fig:pattern1b}$^*$&0.03 &0.06  \\
\cmidrule{1-7}
b. dolomite (31)/&\multirow{2}{*}{negative}&increasing&Fig.~\ref{fig:pattern2}&Fig.~\ref{fig:pattern2b}&0.21 &0.41 \\
\,\,\,b. limestone (32)&&decreasing&Fig.~\ref{fig:pattern1}&Fig.~\ref{fig:pattern1b}$^*$& 0.43&$>1$  \\
\cmidrule{1-7}
g. dolomite (31)/&\multirow{2}{*}{positive}&increasing&Fig.~\ref{fig:pattern1}&Fig.~\ref{fig:pattern1b}$^*$&0.14&0.25  \\
\,\,\,b. limestone (32)&&decreasing&Fig.~\ref{fig:pattern2}&Fig.~\ref{fig:pattern2b}$^*$&0.43 & $>1$ \\
\bottomrule
\end{tabular}}
\label{tab:225}
\end{table}

\begin{table}[!htbp]
\scalebox{0.99}{
\begin{tabular}
{ccccccc}
\cmidrule{2-7}
&\multicolumn{6}{c}{$\theta=30^\circ$}\\
\toprule
 \multirow{2}{*}{models} & $R_{pp}$& $R_{pp}$ mono-&expected & actual &\multirow{2}{*}{$e_{ir}$} & \multirow{2}{*}{$e_{rr}$}\\
  & at $e=0$ &tonicity&pattern & pattern &   & \\
 \cmidrule{1-7}
b. sand (E5)/&\multirow{2}{*}{positive}&non mono.&Fig.~\ref{fig:pattern1}&Fig.~\ref{fig:pattern1b}&$0.09$ & $0.30$ \\
\,\,\,b. sand (E2) && decreasing&Fig.~\ref{fig:pattern2} &Fig.~\ref{fig:pattern2b}&$0.31$ &$>1$ \\
\cmidrule{1-7}
g. sand (E5)/&\multirow{2}{*}{positive}&non mono.&Fig.~\ref{fig:pattern1} &Fig.~\ref{fig:pattern1a}$^*$&--- & --- \\
\,\,\,b. sand (E2) && decreasing&Fig.~\ref{fig:pattern2}  &Fig.~\ref{fig:pattern2b}$^*$&0.31 &$>1$ \\
\cmidrule{1-7}
b. limestone (1)/&\multirow{2}{*}{negative}&increasing&Fig.~\ref{fig:pattern2}&Fig.~\ref{fig:pattern2b}&$0.32$ & $>1$\\
\,\,\,b. limestone (2) &&decreasing& Fig.~\ref{fig:pattern1} &Fig.~\ref{fig:pattern1b}$^*$&$0.26$ &$>1$ \\
\cmidrule{1-7}
g. limestone (1)/&\multirow{2}{*}{positive}&increasing&Fig.~\ref{fig:pattern1} &Fig.~\ref{fig:pattern1b}$^*$ &$0.24$ & $>1$ \\
\,\,\,b. limestone (2) &&decreasing& Fig.~\ref{fig:pattern2}  &Fig.~\ref{fig:pattern2b}&$0.26$ &$>1$ \\
\cmidrule{1-7}
b. shale (B1)/&\multirow{2}{*}{positive}&increasing&Fig.~\ref{fig:pattern1}  &Fig.~\ref{fig:pattern1b}$^*$&0.29 &  $>1$\\
\,\,\,b. shale (B2)&&non mono.&Fig.~\ref{fig:pattern2}  &Fig.~\ref{fig:pattern2b}&0.06&0.19  \\
\cmidrule{1-7}
b. shale (G3)/&\multirow{2}{*}{negative}&non mono.&Fig.~\ref{fig:pattern2} &Fig.~\ref{fig:pattern2b}$^*$ & 0.15& --- \\
\,\,\,b. shale (G5)&&non mono.&Fig.~\ref{fig:pattern1} &none &0.01&0.02 \\
\cmidrule{1-7}
b. shale (E1)/&\multirow{2}{*}{positive}&increasing&Fig.~\ref{fig:pattern1}&Fig.~\ref{fig:pattern1b}$^*$& $>1$&--- \\
\,\,\,b. shale (E5)&&decreasing&Fig.~\ref{fig:pattern2}&Fig.~\ref{fig:pattern2b}$^*$& $>1$&--- \\
\cmidrule{1-7}
b. sand (E5)/&\multirow{2}{*}{positive}&non mono.&Fig.~\ref{fig:pattern1}&Fig.~\ref{fig:pattern1b}& 0.09 &0.30  \\
\,\,\,b. shale (E5)&&decreasing&Fig.~\ref{fig:pattern2}&Fig.~\ref{fig:pattern2b}$^*$& $>1$ & ---\\
\cmidrule{1-7}
g. sand (E5)/&\multirow{2}{*}{positive}&non mono.&Fig.~\ref{fig:pattern1}&Fig.~\ref{fig:pattern1a}$^*$&--- &---  \\
\,\,\,b. shale (E5)&&decreasing&Fig.~\ref{fig:pattern2}&Fig.~\ref{fig:pattern2b}$^*$&$>1$&---  \\
\cmidrule{1-7}
g. sand (G8)/&\multirow{2}{*}{positive}&non mono.&Fig.~\ref{fig:pattern1}&Fig.~\ref{fig:pattern1b}$^*$& 0.02& 0.04 \\
\,\,\,b. sand (G8)&&non mono.&Fig.~\ref{fig:pattern2}&Fig.~\ref{fig:pattern2b}&0.22&$>1$  \\
\cmidrule{1-7}
g. sand (G14)/&\multirow{2}{*}{negative}&non mono.&Fig.~\ref{fig:pattern2}&Fig.~\ref{fig:pattern2b}$^*$&0.03&0.06\\
\,\,\,g. sand (G16)&&decreasing&Fig.~\ref{fig:pattern1}&Fig.~\ref{fig:pattern1b}& 0.01& 0.03 \\
\cmidrule{1-7}
g. coal (G31)/&\multirow{2}{*}{positive}&non mono.&Fig.~\ref{fig:pattern1}&Fig.~\ref{fig:pattern1b}& 0.24&$>1$ \\
\,\,\,b. coal (G31)&&non mono.&Fig.~\ref{fig:pattern2}&Fig.~\ref{fig:pattern2b}$^*$& 0.56 & --- \\
\cmidrule{1-7}
g. limestone (9)/&\multirow{2}{*}{positive}&increasing&Fig.~\ref{fig:pattern1}&Fig.~\ref{fig:pattern1b}$^*$& 0.02& 0.07  \\
\,\,\,g. limestone (10)&&decreasing&Fig.~\ref{fig:pattern2}&Fig.~\ref{fig:pattern2b}$^*$&0.24 &$>1$  \\
\cmidrule{1-7}
b. limestone (9)/&\multirow{2}{*}{positive}&increasing&Fig.~\ref{fig:pattern1}&Fig.~\ref{fig:pattern1b}$^*$&0.24 & $>1$ \\
\,\,\,g. limestone (10)&&decreasing&Fig.~\ref{fig:pattern2}&Fig.~\ref{fig:pattern2b}& 0.24&$>1$  \\
\cmidrule{1-7}
b. limestone (22)/&\multirow{2}{*}{positive}&increasing&Fig.~\ref{fig:pattern1}&Fig.~\ref{fig:pattern1b}$^*$&0.36 &$>1$  \\
\,\,\,b. dolomite (23)&&decreasing&Fig.~\ref{fig:pattern2}&Fig.~\ref{fig:pattern2b}$^*$& 0.24 & $>1$ \\
\cmidrule{1-7}
g. limestone (22)/&\multirow{2}{*}{positive}&increasing&Fig.~\ref{fig:pattern1}&Fig.~\ref{fig:pattern1b}$^*$&0.19 & 0.76  \\
\,\,\,b. dolomite (23)&&non mono.&Fig.~\ref{fig:pattern2}&Fig.~\ref{fig:pattern2b}$^*$& 0.24&$>1$  \\
\cmidrule{1-7}
g. dolomite (28)/&\multirow{2}{*}{negative}&increasing&Fig.~\ref{fig:pattern2}&Fig.~\ref{fig:pattern2b}$^*$& 0.29 & $>1$\\
\,\,\,g. dolomite (29)&&non mono.&Fig.~\ref{fig:pattern1}&Fig.~\ref{fig:pattern1b}&0.11& 0.32  \\
\cmidrule{1-7}
b. dolomite (28)/&\multirow{2}{*}{negative}&non mono.&Fig.~\ref{fig:pattern2}&Fig.~\ref{fig:pattern2b}$^*$& 0.27& $>1$\\
\,\,\,g. dolomite (29)&&decreasing&Fig.~\ref{fig:pattern1}&Fig.~\ref{fig:pattern1b}$^*$&0.11 &0.32  \\
\cmidrule{1-7}
b. dolomite (31)/&\multirow{2}{*}{negative}&increasing&Fig.~\ref{fig:pattern2}&Fig.~\ref{fig:pattern2b}&0.31 &$>1$ \\
\,\,\,b. limestone (32)&&decreasing&Fig.~\ref{fig:pattern1}&Fig.~\ref{fig:pattern1b}$^*$& 0.56&$>1$  \\
\cmidrule{1-7}
g. dolomite (31)/&\multirow{2}{*}{positive}&increasing&Fig.~\ref{fig:pattern1}&Fig.~\ref{fig:pattern1b}$^*$& 0.22& $>1$ \\
\,\,\,b. limestone (32)&&decreasing&Fig.~\ref{fig:pattern2}&Fig.~\ref{fig:pattern2b}$^*$&0.56 & $>1$ \\
\bottomrule
\end{tabular}}
\label{tab:30}
\end{table}

\begin{table}[!htbp]
\scalebox{0.99}{
\begin{tabular}
{ccccccc}
\cmidrule{2-7}
&\multicolumn{6}{c}{$\theta=37.5^\circ$}\\
\toprule
 \multirow{2}{*}{models} & $R_{pp}$&$R_{pp}$ mono-&expected & actual &\multirow{2}{*}{$e_{ir}$} & \multirow{2}{*}{$e_{rr}$}\\
  & at $e=0$ &tonicity&pattern & pattern &   & \\
 \cmidrule{1-7}
b. sand (E5)/&\multirow{2}{*}{positive}&non mono.&Fig.~\ref{fig:pattern1b}&Fig.~\ref{fig:pattern1b}$^*$&0.19 & --- \\
\,\,\,b. sand (E2) &&decreasing& Fig.~\ref{fig:pattern2b} &Fig.~\ref{fig:pattern2b}$^*$&$0.45$ &--- \\
\cmidrule{1-7}
g. sand (E5)/&\multirow{2}{*}{positive}&non mono.&Fig.~\ref{fig:pattern1b}&Fig.~\ref{fig:pattern1b}&0.04 & 0.23 \\
\,\,\,b. sand (E2) &&decreasing& Fig.~\ref{fig:pattern2b} &Fig.~\ref{fig:pattern2b}$^*$&0.45 &--- \\
\cmidrule{1-7}
b. limestone (1)/&\multirow{2}{*}{negative}&increasing&Fig.~\ref{fig:pattern2b}&Fig.~\ref{fig:pattern2b}$^*$&$0.47$ & --- \\
\,\,\,b. limestone (2) &&decreasing&Fig.~\ref{fig:pattern1b}&Fig.~\ref{fig:pattern1b}$^*$&$0.40$ &--- \\
\cmidrule{1-7}
g. limestone (1)/&\multirow{2}{*}{positive}&increasing&Fig.~\ref{fig:pattern1b}&Fig.~\ref{fig:pattern1b}$^*$&$0.37$ & --- \\
\,\,\,b. limestone (2) &&decreasing&Fig.~\ref{fig:pattern2b} &Fig.~\ref{fig:pattern2b}$^*$&$0.40$ &---\\
\cmidrule{1-7}
b. shale (B1)/&\multirow{2}{*}{positive}&increasing&Fig.~\ref{fig:pattern1b} &Fig.~\ref{fig:pattern1b}&0.43& ---\\
\,\,\,b. shale (B2)&& non mono.&Fig.~\ref{fig:pattern2b} &Fig.~\ref{fig:pattern2b}$^*$&0.16&---  \\
\cmidrule{1-7}
b. shale (G3)/&\multirow{2}{*}{negative}&non mono.&Fig.~\ref{fig:pattern2b} &Fig.~\ref{fig:pattern2b}$^*$& 0.27& --- \\
\,\,\,b. shale (G5)&&non mono.&Fig.~\ref{fig:pattern1b} &Fig.~\ref{fig:pattern1b}&0.10&$>1$ \\
\cmidrule{1-7}
b. shale (E1)/&\multirow{2}{*}{positive}&increasing&Fig.~\ref{fig:pattern1b}&Fig.~\ref{fig:pattern1b}& $>1$&---\\
\,\,\,b. shale (E5)&&decreasing&Fig.~\ref{fig:pattern2b}&Fig.~\ref{fig:pattern2b}$^*$& $>1$&--- \\
\cmidrule{1-7}
b. sand (E5)/&\multirow{2}{*}{positive}&non mono.&Fig.~\ref{fig:pattern1b}&Fig.~\ref{fig:pattern1b}$^*$&0.19&--- \\
\,\,\,b. shale (E5)&&decreasing&Fig.~\ref{fig:pattern2b}&Fig.~\ref{fig:pattern2b}$^*$&$>1$ & ---\\
\cmidrule{1-7}
g. sand (E5)/&\multirow{2}{*}{positive}&non mono.&Fig.~\ref{fig:pattern1b}&Fig.~\ref{fig:pattern1b}$^*$&0.04 &0.23  \\
\,\,\,b. shale (E5)&&decreasing&Fig.~\ref{fig:pattern2b}&Fig.~\ref{fig:pattern2b}$^*$&$>1$&--- \\
\cmidrule{1-7}
g. sand (G8)/&\multirow{2}{*}{positive}&increasing&Fig.~\ref{fig:pattern1b}&Fig.~\ref{fig:pattern1b}& 0.11& $>1$\\
\,\,\,b. sand (G8)&&non mono.&Fig.~\ref{fig:pattern2b}&Fig.~\ref{fig:pattern2b}$^*$&0.35&---\\
\cmidrule{1-7}
g. sand (G14)/&\multirow{2}{*}{negative}&non mono.&Fig.~\ref{fig:pattern2b}&Fig.~\ref{fig:pattern2b}$^*$&0.12&$>1$\\
\,\,\,g. sand (G16)&&decreasing&Fig.~\ref{fig:pattern1b}&Fig.~\ref{fig:pattern1b}& 0.11& $>1$ \\
\cmidrule{1-7}
g. coal (G31)/&\multirow{2}{*}{positive}&non mono.&Fig.~\ref{fig:pattern1b}&Fig.~\ref{fig:pattern1b}& 0.41&---  \\
\,\,\,b. coal (G31)&&non mono.&Fig.~\ref{fig:pattern2b}&Fig.~\ref{fig:pattern2b}$^*$& 0.78& --- \\
\cmidrule{1-7}
g. limestone (9)/&\multirow{2}{*}{positive}&increasing&Fig.~\ref{fig:pattern1b}&Fig.~\ref{fig:pattern1b}$^*$&0.11&$>1$ \\
\,\,\,g. limestone (10)&&decreasing&Fig.~\ref{fig:pattern2b}&Fig.~\ref{fig:pattern2b}$^*$&0.38 &---  \\
\cmidrule{1-7}
b. limestone (9)/&\multirow{2}{*}{positive}&increasing&Fig.~\ref{fig:pattern1b}&Fig.~\ref{fig:pattern1b}$^*$&0.38 &--- \\
\,\,\,g. limestone (10)&&decreasing&Fig.~\ref{fig:pattern2b}&Fig.~\ref{fig:pattern2b}$^*$&0.38&---  \\
\cmidrule{1-7}
b. limestone (22)/&\multirow{2}{*}{positive}&increasing&Fig.~\ref{fig:pattern1b}&Fig.~\ref{fig:pattern1b}$^*$& 0.52& --- \\
\,\,\,b. dolomite (23)&&decreasing&Fig.~\ref{fig:pattern2b}&Fig.~\ref{fig:pattern2b}$^*$&0.38  &---  \\
\cmidrule{1-7}
g. limestone (22)/&\multirow{2}{*}{positive}&increasing&Fig.~\ref{fig:pattern1b}&Fig.~\ref{fig:pattern1b}$^*$& 0.31&---   \\
\,\,\,b. dolomite (23)&&decreasing&Fig.~\ref{fig:pattern2b}&Fig.~\ref{fig:pattern2b}$^*$& 0.38&---  \\
\cmidrule{1-7}
g. dolomite (28)/&\multirow{2}{*}{negative}&increasing&Fig.~\ref{fig:pattern2b}&Fig.~\ref{fig:pattern2b}$^*$& 0.44& ---\\
\,\,\,g. dolomite (29)&&decreasing&Fig.~\ref{fig:pattern1b}&Fig.~\ref{fig:pattern1b}$^*$&0.21& ---  \\
\cmidrule{1-7}
b. dolomite (28)/&\multirow{2}{*}{negative}&increasing&Fig.~\ref{fig:pattern2b}&Fig.~\ref{fig:pattern2b}$^*$& 0.41& ---\\
\,\,\,g. dolomite (29)&&decreasing& Fig.~\ref{fig:pattern1b}&Fig.~\ref{fig:pattern1b}$^*$&0.21 &---  \\
\cmidrule{1-7}
b. dolomite (31)/&\multirow{2}{*}{positive}&increasing&Fig.~\ref{fig:pattern1b}&Fig.~\ref{fig:pattern1b}&0.46&--- \\
\,\,\,b. limestone (32)&&decreasing&Fig.~\ref{fig:pattern2b}&Fig.~\ref{fig:pattern2b}$^*$& 0.76&---  \\
\cmidrule{1-7}
g. dolomite (31)/&\multirow{2}{*}{positive}&increasing&Fig.~\ref{fig:pattern1b}&Fig.~\ref{fig:pattern1b}$^*$&0.35&---  \\
\,\,\,b. limestone (32)&&decreasing&Fig.~\ref{fig:pattern2b}&Fig.~\ref{fig:pattern2b}$^*$&0.76 & --- \\
\bottomrule
\end{tabular}}
\label{tab:375}
\end{table}

\begin{table}[!htbp]
\scalebox{0.99}{
\begin{tabular}
{ccccccc}
\cmidrule{2-7}
&\multicolumn{6}{c}{$\theta=45^\circ$}\\
\toprule
 \multirow{2}{*}{models} & $R_{pp}$& $R_{pp}$ mono-&expected & actual &\multirow{2}{*}{$e_{ir}$} & \multirow{2}{*}{$e_{rr}$}\\
  & at $e=0$ &tonicity&pattern & pattern &   & \\
 \cmidrule{1-7}
b. sand (E5)/& \multirow{2}{*}{positive}&increasing&Fig.~\ref{fig:pattern1b}&Fig.~\ref{fig:pattern1b}&$0.35$ & --- \\
\,\,\,b. sand (E2) &&decreasing& Fig.~\ref{fig:pattern2b} &Fig.~\ref{fig:pattern2b}&$0.69$ &$>1$ \\
\cmidrule{1-7}
g. sand (E5)/&\multirow{2}{*}{positive}&increasing&Fig.~\ref{fig:pattern1b}&Fig.~\ref{fig:pattern1b}$^*$&0.21 & --- \\
\,\,\,b. sand (E2) &&decreasing& Fig.~\ref{fig:pattern2b} &Fig.~\ref{fig:pattern2b}$^*$&0.69 &--- \\
\cmidrule{1-7}
b. limestone (1)/&\multirow{2}{*}{negative}&increasing&Fig.~\ref{fig:pattern2b} &Fig.~\ref{fig:pattern2b}$^*$&$0.71$ & --- \\
\,\,\,b. limestone (2) &&decreasing&Fig.~\ref{fig:pattern1b}&Fig.~\ref{fig:pattern1b}&$0.62$ &--- \\
\cmidrule{1-7}
g. limestone (1)/&\multirow{2}{*}{positive}&increasing&Fig.~\ref{fig:pattern1b} &Fig.~\ref{fig:pattern1b}&$0.58$ & --- \\
\,\,\,b. limestone (2) &&decreasing& Fig.~\ref{fig:pattern2b}  &Fig.~\ref{fig:pattern2b}$^*$&$0.62$ &--- \\
\cmidrule{1-7}
b. shale (B1)/&\multirow{2}{*}{negative}&increasing&Fig.~\ref{fig:pattern2b}&Fig.~\ref{fig:pattern2b}$^*$&0.66 & ---\\
\,\,\,b. shale (B2)&&non mono.&Fig.~\ref{fig:pattern1b}&Fig.~\ref{fig:pattern1b}&0.32&---  \\
\cmidrule{1-7}
b. shale (G3)/&\multirow{2}{*}{negative}&increasing&Fig.~\ref{fig:pattern2b}&Fig.~\ref{fig:pattern2b}$^*$&0.46& --- \\
\,\,\,b. shale (G5)&&non mono.&Fig.~\ref{fig:pattern1b}&Fig.~\ref{fig:pattern1b}&0.24&--- \\
\cmidrule{1-7}
b. shale (E1)/&\multirow{2}{*}{positive}&increasing&Fig.~\ref{fig:pattern1b}& Fig.~\ref{fig:pattern1b}&$>1$&--- \\
\,\,\,b. shale (E5)&&decreasing&Fig.~\ref{fig:pattern2b}&Fig.~\ref{fig:pattern2b}$^*$& $>1$&--- \\
\cmidrule{1-7}
b. sand (E5)/&\multirow{2}{*}{positive}&non mono.&Fig.~\ref{fig:pattern1b}&Fig.~\ref{fig:pattern1b}$^*$ & 0.35 &---\\
\,\,\,b. shale (E5)&&decreasing&Fig.~\ref{fig:pattern2b}&Fig.~\ref{fig:pattern2b}$^*$& $>1$ & ---\\
\cmidrule{1-7}
g. sand (E5)/&\multirow{2}{*}{positive}&non mono.&Fig.~\ref{fig:pattern1b}&Fig.~\ref{fig:pattern1b}$^*$&0.15&---  \\
\,\,\,b. shale (E5)&&decreasing&Fig.~\ref{fig:pattern2b}&Fig.~\ref{fig:pattern2b}$^*$&$>1$&---  \\
\cmidrule{1-7}
g. sand (G8)/&\multirow{2}{*}{positive}&increasing&Fig.~\ref{fig:pattern1b}&Fig.~\ref{fig:pattern1b}$^*$& 0.24& --- \\
\,\,\,b. sand (G8)&&decreasing&Fig.~\ref{fig:pattern2b}&Fig.~\ref{fig:pattern2b}$^*$&0.55&---  \\
\cmidrule{1-7}
g. sand (G14)/&\multirow{2}{*}{negative}&non mono.&Fig.~\ref{fig:pattern2b}&Fig.~\ref{fig:pattern2b}$^*$&0.26&---\\
\,\,\,g. sand (G16)&&decreasing&Fig.~\ref{fig:pattern1b}&Fig.~\ref{fig:pattern1b}& 0.24& --- \\
\cmidrule{1-7}
g. coal (G31)/&\multirow{2}{*}{positive}&increasing&Fig.~\ref{fig:pattern1b}&Fig.~\ref{fig:pattern1b}& 0.66&---  \\
\,\,\,b. coal (G31)&&decreasing&Fig.~\ref{fig:pattern2b}&Fig.~\ref{fig:pattern2b}$^*$& $>1$ & --- \\
\cmidrule{1-7}
g. limestone (9)/&\multirow{2}{*}{positive}&increasing&Fig.~\ref{fig:pattern1b}&Fig.~\ref{fig:pattern1b}$^*$& 0.25& ---  \\
\,\,\,g. limestone (10)&&decreasing&Fig.~\ref{fig:pattern2b}&Fig.~\ref{fig:pattern2b}$^*$&0.59&---  \\
\cmidrule{1-7}
b. limestone (9)/&\multirow{2}{*}{positive}&increasing&Fig.~\ref{fig:pattern1b}&Fig.~\ref{fig:pattern1b}&0.60 & --- \\
\,\,\,g. limestone (10)&&decreasing&Fig.~\ref{fig:pattern2b}&Fig.~\ref{fig:pattern2b}$^*$& 0.59&---  \\
\cmidrule{1-7}
b. limestone (22)/&\multirow{2}{*}{positive}&increasing&Fig.~\ref{fig:pattern1b}&Fig.~\ref{fig:pattern1b}&0.78 &---  \\
\,\,\,b. dolomite (23)&&decreasing&Fig.~\ref{fig:pattern2b}&Fig.~\ref{fig:pattern2b}$^*$& 0.60 & ---\\
\cmidrule{1-7}
g. limestone (22)/&\multirow{2}{*}{positive}&increasing&Fig.~\ref{fig:pattern1b}&Fig.~\ref{fig:pattern1b}$^*$& 0.50&---\\
\,\,\,b. dolomite (23)&&decreasing&Fig.~\ref{fig:pattern2b}&Fig.~\ref{fig:pattern2b}$^*$& 0.60 & ---  \\
\cmidrule{1-7}
g. dolomite (28)/&\multirow{2}{*}{negative}&increasing&Fig.~\ref{fig:pattern2b}&Fig.~\ref{fig:pattern2b}$^*$& 0.68 &--- \\
\,\,\,g. dolomite (29)&&decreasing&Fig.~\ref{fig:pattern1b}&Fig.~\ref{fig:pattern1b}&0.39& ---  \\
\cmidrule{1-7}
b. dolomite (28)/&\multirow{2}{*}{negative}&increasing&Fig.~\ref{fig:pattern2b}&Fig.~\ref{fig:pattern2b}$^*$ &0.63&---  \\
\,\,\,g. dolomite (29)&&decreasing&Fig.~\ref{fig:pattern1b}&Fig.~\ref{fig:pattern1b}$^*$&0.39 &---  \\
\cmidrule{1-7}
b. dolomite (31)/&\multirow{2}{*}{positive}&increasing&Fig.~\ref{fig:pattern1b}&Fig.~\ref{fig:pattern1b}&0.71 &--- \\
\,\,\,b. limestone (32)&&decreasing&Fig.~\ref{fig:pattern2b}&Fig.~\ref{fig:pattern2b}$^*$& $>1$&---  \\
\cmidrule{1-7}
g. dolomite (31)/&\multirow{2}{*}{positive}&increasing&Fig.~\ref{fig:pattern1b}&Fig.~\ref{fig:pattern1b}$^*$&0.56&--- \\
\,\,\,b. limestone (32)&&decreasing&Fig.~\ref{fig:pattern2b}&Fig.~\ref{fig:pattern2b}$^*$&$>1$ & --- \\
\bottomrule
\end{tabular}}
\label{tab:45}
\end{table}

\end{appendix}
\newpage
\section{Matlab code}\label{ap:ch8_two}
\lstset{language=Matlab,%
    breaklines=true,%
    morekeywords={matlab2tikz},
    keywordstyle=\color{blue},%
    morekeywords=[2]{1}, keywordstyle=[2]{\color{black}},
    identifierstyle=\color{black},%
    stringstyle=\color{mylilas},
    commentstyle=\color{mygreen},%
    showstringspaces=false,
    numbers=none,%
    numberstyle={\tiny \color{black}},
    numbersep=9pt, 
    emph=[1]{for,end,break},emphstyle=[1]\color{red}, 
}

\begin{lstlisting}
%  This code computes azimuthaly-dependent reflection coefficient using Vavrycuk-Psencik approximation.     Both VTI halfspaces may have vertical cracks aligned along the x1-axis. Code uses matlab functions.
%%% INPUTS %%%
% e    - crack density of an upper halfspace
% e_L  - crack density of a lower halfspace
% C##b - background stiffnesses of a VTI (or isotropic) upper halfspace
% L##b - background stiffnesses of a VTI (or isotropic) lower halfspace
% ro   - density of an upper halfspace
% ro_L - density of a lower halfspace
% x    - incidence angle
% y    - azimuthal angle
%%% OUTPUTS %%%
% C##  - effective stiffnesses of an upper halfspace
% L##  - effective stiffnesses of a lower halfspace
% Ripp - isotropic part of the PP refl. coeff.
% Rpp  - Vavrycuk-Psencik approximated PP refl. coeff.
%%% MAIN CODE %%%
[C11,C22,C33,C44,C55,C66,C12,C13,C23,L11,L22,L33,L44,L55,L66,L12,L13,L23]=fun_eff(e,e_L,C11b,C33b,C44b,C66b,C13b,L11b,L33b,L44b,L66b,L13b);
[Ripp,Rpp]=fun_Rpp_approx(x,y,ro,ro_L,C11,C22,C33,C44,C55,C66,C12,C13,C23,L11,L22,L33,L44,L55,L66,L12,L13,L23);
%%% FUNCTIONS %%%
function [C11,C22,C33,C44,C55,C66,C12,C13,C23,L11,L22,L33,L44,L55,L66,L12,L13,L23]=fun_eff(e,e_L,C11b,C33b,C44b,C66b,C13b,L11b,L33b,L44b,L66b,L13b)
% obtain effective stiffnesses of an upper halfspace
C1=sqrt(C11b*C33b);
C2=sqrt(C66b/C44b);
C3=sqrt(((C1-C13b)*(C1+C13b+2*C44b))/(C33b*C44b));
C4=2*C44b*C3/(C1+C13b+2*C44b);
ZN=8*C3*e/(3*C1*(1-((C13b^2)/(C1^2))));
ZT=16*e/(3*C44b*(C2+C3-C4));
delN=ZN*C11b/(1+ZN*C11b);
delT1=ZT*C44b/(1+ZT*C44b);
delT2=ZT*C66b/(1+ZT*C66b);
C11=C11b*(1-delN);
C22=C11b*(1-delN*(((C11b-2*C66b)^2)/(C11b^2)));
C33=C33b*(1-delN*((C13b^2)/(C11b*C33b)));
C12=(C11b-2*C66b)*(1-delN);
C13=C13b*(1-delN);
C23=C13b*(1-delN*(((C11b-2*C66b))/C11b));
C44=C44b; C55=C44b.*(1-delT1); C66=C66b.*(1-delT2);
% obtain effective stiffnesses of a lower halfspace
L1=sqrt(L11b*L33b);
L2=sqrt(L66b/L44b);
L3=sqrt(((L1-L13b)*(L1+L13b+2*L44b))/(L33b*L44b));
L4=2*L44b*L3/(L1+L13b+2*L44b);
ZN_L=8*L3*e_L/(3*L1*(1-((L13b^2)/(L1^2))));
ZT_L=16*e_L/(3*L44b.*(L2+L3-L4));
delN_L=ZN_L*L11b/(1+ZN_L*L11b);
delT1_L=ZT_L*L44b/(1+ZT_L*L44b);
delT2_L=ZT_L*L66b/(1+ZT_L*L66b);
L11=L11b*(1-delN_L);
L22=L11b*(1-delN_L*(((L11b-2*L66b)^2)/(L11b^2)));
L33=L33b*(1-delN_L*((L13b^2)/(L11b*L33b)));
L12=(L11b-2*L66b)*(1-delN_L);
L13=L13b*(1-delN_L);
L23=L13b*(1-delN_L*(((L11b-2*L66b))/L11b));
L44=L44b; L55=L44b.*(1-delT1_L); L66=L66b.*(1-delT2_L);
end
function [Ripp,Rpp]=fun_Rpp_approx(x,y,ro,ro_L,C11,C22,C33,C44,C55,C66,C12,C13,C23,L11,L22,L33,L44,L55,L66,L12,L13,L23);
% obtain the PP refl. coeff. approximation using effective stiffnesses
Vp=sqrt((C33)/ro); Vp_L=sqrt((L33)/ro_L);
Vs=sqrt((C55)/ro); Vs_L=sqrt((L55)/ro_L);
mVp=0.5*(Vp_L+Vp); mVs=0.5*(Vs_L+Vs); dVp=Vp_L-Vp; 
Z=ro*Vp; Z_L=ro_L*Vp_L; dZ=Z_L-Z; mZ=0.5*(Z_L+Z);
G=ro*Vs^2; G_L=ro_L*Vs_L^2; dG=G_L-G; mG=0.5*(G_L+G);
for j=1:length(x)
for i=1:length(y)
Ripp(j)=0.5.*dZ./mZ+0.5.*(dVp./mVp).*tand(x(j)).^2-2.*(dG./mG).*(mVs./mVp).^2.*sind(x(j)).^2;
Rapp1(i)=cosd(y(i)).^2.*(((L13+2.*L55-L33)./L33)-((C13+2.*C55-C33)./C33))+sind(y(i)).^2.*(((((L23+2.*L44-L33)./L33)-((C23+2.*C44-C33)./C33)))-8.*(((L44-L55)./(2.*L33))-((C44-C55)./(2.*C33))));
Rapp2(i)=cosd(y(i)).^4.*(((L11-L33)./(2.*L33))-((C11-C33)./(2.*C33)))+sind(y(i)).^4.*(((L22-L33)./(2.*L33))-((C22-C33)./(2.*C33)))+sind(y(i)).^2.*cosd(y(i)).^2.*(((L12+2.*L66-L33)./L33)-((C12+2.*C66-C33)./C33));
Rpp(j,i)=Ripp(j)+0.5.*sind(x(j)).^2.*Rapp1(i)+0.5.*sind(x(j)).^2.*tand(x(j)).^2.*Rapp2(i);
end; end; end
\end{lstlisting}
\end{document}